\documentclass{article}
\usepackage[utf8]{inputenc}

\usepackage{amsmath,amssymb,bbm,amsthm}
\usepackage{fullpage}
\usepackage{thm-restate,color,xcolor,xspace}
\usepackage{hyperref,cleveref}
\usepackage{algorithm,algorithmic}
\usepackage{thm-restate}
\usepackage{graphicx}
\usepackage{titlesec}
\usepackage{MnSymbol} 

\newtheorem{theorem}{Theorem}[section]
\newtheorem*{theorem*}{Theorem}
\newtheorem{lemma}[theorem]{Lemma}

\newtheorem{definition}[theorem]{Definition}
\newtheorem{corollary}[theorem]{Corollary}
\newtheorem{remark}[theorem]{Remark}

\newtheorem{observation}[theorem]{Observation}

\newtheorem{problem}[theorem]{Problem}

\newtheorem{conjecture}[theorem]{Conjecture}

\title{Algorithmic Applications of Hypergraph and Partition Containers}
\author{Or Zamir\\Princeton University}
\date{}

\begin{document}

\maketitle

\begin{abstract}
We present a general method to convert algorithms into faster algorithms for  \emph{almost-regular} input instances.
Informally, an almost-regular input is an input in which the maximum degree is larger than the average degree by at most a constant factor.
This family of inputs vastly generalizes several families of inputs for which we commonly have improved algorithms, including bounded-degree inputs and random inputs. It also generalizes families of inputs for which we don't usually have faster algorithms, including regular-inputs of arbitrarily  high degree and very dense inputs.
We apply our method to achieve breakthroughs in exact algorithms for several central NP-Complete problems including~$k$-SAT, Graph Coloring, and Maximum Independent Set.

Our main tool is the first algorithmic application of the relatively new Hypergraph Container Method (Saxton and Thomason~\cite{saxton2015hypergraph}, Balogh, Morris and Samotij \cite{balogh2015independent}). 
This recent breakthrough, which generalizes an earlier version for graphs (Kleitman and Winston~\cite{kleitman1982number}, Sapozhenko~\cite{sapozhenko2001number}), has been used extensively in recent years in extremal combinatorics.
An important component of our work is the generalization of (hyper-)graph containers to~\emph{Partition Containers}.
\end{abstract}

\thispagestyle{empty}
\clearpage
\setcounter{page}{1}

\tableofcontents
\clearpage

\section{Introduction}
For many problems, we can design faster algorithms if the inputs are of some restricted form.
Examples of common families of inputs for which significantly better algorithms are known include sparse or bounded-degree inputs, and random inputs.
In this paper we present a general approach to devise improved algorithms for the much broader family of \emph{almost-regular} input instances. 
For example, we say that a graph is \emph{almost-regular} if its maximum degree is at most~$C$ times larger than its average degree, for some fixed constant~$C$.
This includes bounded-degree and random inputs, but also regular-graphs of an arbitrarily high degree and \emph{all} very dense inputs; If a graph on~$n$ vertices contains at least~$\varepsilon n^2$ edges, then it is almost-regular with~$C=\frac{1}{2\varepsilon}$.
Using our approach, we prove that many problems are inherently easier to solve for very general families of inputs. We cover several problems for which similar results were not known before for any interesting family of inputs, or were known only for much more restricted inputs.

Our first main application is a resolution of a major open problem about Graph Coloring algorithms, for almost-regular graphs.
The chromatic number of a graph can be computed in~$O^*\left(2^n\right)$ time (Bj\"{o}rklund, Husfeldt and Koivisto~\cite{bjorklund2009set}). 
For~$k\leq 6$ it is known that~$k$-coloring can be solved in~$O\left(\left(2-\varepsilon\right)^n\right)$ time for some~$\varepsilon>0$ (Biegel and Eppstein~\cite{beigel20053}, Fomin, Gaspers and Saurabh~\cite{fomin2007improved}, Zamir~\cite{DBLP:conf/icalp/Zamir21}). 
For larger values of~$k$ improvements were only known for sparse graphs (Zamir~\cite{DBLP:conf/icalp/Zamir21}).
We solve~$k$-coloring in~$O\left(\left(2-\varepsilon\right)^n\right)$ time for graphs in which the maximum degree is at most~$C$ times the average degree, where~$\varepsilon = \varepsilon_{k,C}>0$ for \emph{every}~$k,C$.
This includes, for example, regular graphs \emph{of any degree} (with~$C=1$), and graphs with at least~$\varepsilon n^2$ edges for any fixed constant~$\varepsilon$ (with~$C=\frac{1}{2\varepsilon}$).

Our second main application is the first improved $k$-SAT algorithm for \emph{dense} formulas.
The celebrated sparsification lemma (Impagliazzo, Paturi and Zane~\cite{impagliazzo2001problems}) states that for every~$k,\varepsilon$ any~$k$-SAT formula on~$n$ variables can be reduced to a disjunction of~$2^{\varepsilon n}$ $k$-SAT formulas on~$n$ variables and at most~$C_{k,\varepsilon} \cdot n$ clauses. In particular, sparse~$k$-SAT formulas are at least as hard to solve as general~$k$-SAT formulas.
%Improved algorithms are known for very sparse formulas (cite).
We show a complementing statement; For every~$k,b$ there exists a~$C$ such that if~$k$-SAT can be solved in~$b^n$ time, then~$k$-SAT on formulas with at least~$C\cdot n$ clauses that are \emph{well-spread} can be solved in~$\left(b-\varepsilon\right)^n$ time.

A few of our results use the algorithms for the unrestricted case in a~\emph{black-box} manner.
For example, the second result above states that if~$k$-SAT can be solved in~$c^n$ time, then~$k$-SAT can be solved in~$(c-\varepsilon)^n$ time for dense inputs, \emph{regardless of what~$c>1$ is}. 
We also demonstrate a white-box use of our approach for the graph coloring problem.

At the heart of our approach lies the celebrated Hypergraph Container Method.
Relatively recently, Saxton and Thomason~\cite{saxton2015hypergraph} and independently Balogh, Morris and Samotij~\cite{balogh2015independent}, developed a powerful method to characterize the structure of independent sets in~\emph{nice} hypergraphs.
This result led to many interesting consequences including several tight counting results and random sparse analogs of classical results in extremal combinatorics. Many of these applications can be found in~\cite{saxton2015hypergraph, balogh2015independent, balogh2018method}.
The basic approach for graphs (instead of hypergraphs) was discovered earlier by several researchers, most notably by Sapozhenko~\cite{sapozhenko2001number} and by Kleitman and Winston~\cite{kleitman1982number}.
In essence, they show that given a hypergraph that satisfies some natural regularity-type conditions, every independent set of it must be fully contained in one of a small number of \emph{somewhat-small} sets (called \emph{containers}).
We next give a very informal statement of the hypergraph container theorem.
\begin{theorem*}[Very informal]
    Let~$\mathcal{H}$ be a hypergraph satisfying certain conditions.
    Then, there exists a collection~$\mathcal{C}$ of subsets of vertices~$C_1,\ldots,C_r\subseteq V(\mathcal{H})$ for which the following hold.
    \begin{itemize}
        \item The number of containers is small, that is~$r=2^{o(|V(\mathcal{H})|)}$.
        \item Each container is small, that is~$|C_i|\leq (1-\varepsilon)|V(\mathcal{H})|$ for every~$i\in [r]$ and some constant~$\varepsilon>0$.
        \item Every independent set~$I\in \mathcal{I}(\mathcal{H})$ in~$\mathcal{H}$ is fully contained in some container, that is~$I\subseteq C_i$ for some~$i\in [r]$.
    \end{itemize}
\end{theorem*}
The container lemma for graphs is precisely presented and proved in Section~\ref{sec:graphconts}, and the one for hypergraphs is precisely presented in Section~\ref{subsec:hyperconts}.

Our main contribution is showing how the hypergraph container method can be used algorithmically.
The simplistic high-level idea would be to algorithmically generate the set of containers~$\mathcal{C}$ for an appropriate (hyper)graph, and then enumerate over all containers while solving the smaller sub-problem we get by restricting the problem to a specific container.
For example, to find the Maximum Independent Set in a graph it is enough to find the Maximum Independent Set in the induced subgraph on each container, and those are much smaller graphs.
An independent set in an induced sub-graph is also an independent set in the entire graph, and the maximum independent set is fully contained in some~$C_i$. Thus, we find the correct maximum independent set. To find the maximum independent set in each container, we can use any maximum independent set algorithm in a black-box manner.

It is sometimes necessary to consider several independent sets simultaneously and thus we cannot restrict ourselves to a single container. For example, this is the case in the graph coloring problem.
We thus generalize the (hyper-)graph container method to something we call \emph{partition containers}, those characterize the structure of several independent sets at the same time.
We further discuss this generalization in Section~\ref{subsec:parts} after introducing our applications.

The concrete applications we present in this paper are all for exponential-time worst-case algorithms for NP-complete problems.
There is a substantial body of work on this type of problems that was developed extensively in the last several decades. See for example the survey of Woeginger~\cite{woeginger2003exact}.
In fact, better-than-enumeration algorithms for such problems appeared at least a decade before the definition of NP (e.g., the Held-Karp algorithm for TSP~\cite{held1962dynamic}).
We demonstrate our approach for several of the most fundamental NP-hard problems: Maximum Independent Set, Graph Coloring, and Satisfiability.

Classically, exact algorithms for NP-Complete problems gained interest for two main reasons. 
First, solutions for those were sometimes required in practice.
Second, the theoretical understanding of the running time of NP-Complete problems is very lacking. 
In fact, we can't even prove that SAT requires super-linear time! 
Thus, we should further our understanding of the landscape of hardness within NP.
More recently, new reasons came up to focus on the exponential running times of these algorithms.
It was discovered that conjectures about the exact running times of NP-complete problems imply that many problems in P cannot be solved polynomially faster than our current solutions (see for example the survey of Vassilevska-Williams on fine-grained complexity \cite{williams2018some}).
The most popular such conjecture is about the running times required to solve~$k$-SAT (SETH \cite{impagliazzo2001complexity}).
Hundreds of results conditioned on these conjectures were published in the last few years.
Hence, focus on exact algorithms for NP-complete problems is completely essential to either further base or disprove these popular conjectures.

\subsection{Almost-regular Maximum Independent Set}\label{subsec:ind}
As a simple example of our approach we discuss algorithms for the unweighted and weighted Maximum Independent Set (MIS) problems.
The first nontrivial algorithm for MIS, running in time~$O^\star(2^{n/3})$, dates back to 1977 by Tarjan and Trojanowski~\cite{tarjan1977finding}.
We use the notation~$O^*(\cdot)$ to hide polynomial factors.
Algorithms for MIS were then improved many times~\cite{jian19862, robson1986algorithms, fomin2009measure, kneis2009fine, bourgeois2012fast, xiao2017exact}.
A lot of attention was also spent on MIS algorithms for bounded-degree graphs~(e.g., \cite{furer2006faster, razgon2009faster, xiao2010simple}).
In fact, many of the general MIS algorithms directly use these bounded-degree case algorithms.
The current state-of-the-art bound is~$1.1996^n$ due to Xiao and Nagamochi~\cite{xiao2017exact}.

In Section~\ref{sec:graphconts}, we use this problem to demonstrate the algorithmic power of containers. Using containers we get faster algorithms for \textbf{large-degree} regular or almost-regular graphs, using the general-case algorithm in a black-box manner.
\begin{theorem*}[\ref{thm:MISreglarge}]
    Given an algorithm that solves MIS in~$O(c^n)$ time, we can solve MIS in~$d$-regular graphs in~$O\left({\sqrt{c}}^{\left(1+o_d(1)\right)n}\right)$ time.
\end{theorem*}
\begin{theorem*}[\ref{thm:MISalmostreglarge}]
    For any~$C>1$ there exists~$\varepsilon=\varepsilon_C>0$ such the the following holds.
    Given an algorithm that solves MIS in~$O(c^n)$ time, we can solve MIS in graphs with average degree~$d$ and maximum degree bounded by~$Cd$ in~$O\left(c^{\left(1-\varepsilon+o_d(1)\right)n}\right)$ time.
\end{theorem*}

Our general ``recipe" then is as follows. Using containers, we get better algorithms for almost-regular instances when the degree is large enough. Otherwise, all degrees are small and we can use different algorithms that are faster in the bounded-degree case.
Together, the combination of using containers and the existence of an improved algorithm for instances with bounded-degrees, leads to an improved algorithm for almost-regular instances without additional assumptions.

\subsection{Almost-regular Graph Coloring}\label{subsec:color}
The problem of $k$-coloring a graph, or determining the \emph{chromatic number} of a graph (i.e., finding the smallest $k$ for which the graph is $k$-colorable) is another one of the most well studied NP-complete problems.
Computing the chromatic number is listed as one of the first NP-complete problems in Karp's paper from 1972~\cite{karp1972reducibility}.
In a similar fashion to $k$-SAT, the problem of $2$-coloring is polynomial, yet $k$-coloring is NP-complete for every $k\geq 3$ (proven independently by Lov{\'a}sz \cite{lovasz1973coverings} and Stockmeyer \cite{stockmeyer1973planar}).

The trivial algorithm solving $k$-coloring by enumerating over all possible colorings takes $O^*(k^n)$ time.
Thus, it is not even immediately clear that computing the chromatic number of a graph can be done in $O^*(c^n)$ time for a constant $c$ independent of $k$.
Nevertheless, in 1976 Lawler \cite{lawler1976note} noted a simple dynamic programming algorithm that computes the chromatic number in $O^*(3^n)$ time. 
More sophisticated algorithms were introduced~\cite{moon1965cliques,paull1959minimizing,eppstein2001small,byskov2004enumerating}, until finally an algorithm computing the chromatic number in $O^*(2^n)$ time was devised by Bj{\"o}rklund, Husfeldt and Koivisto in 2009 \cite{bjorklund2009set}. This settled an open problem of Woeginger \cite{woeginger2003exact}.

For very small values of~$k$, better algorithms are known for the $k$-coloring problem.
Schiermeyer \cite{schiermeyer1993deciding} showed that $3$-coloring can be solved in $O^*(1.415^n)$ time.
Since then, these were improved several times; Currently the best known running times for~$3$-coloring and~$4$-coloring respectively are $O^*(1.3289^n)$ and $O^*(1.7272^n)$~\cite{schiermeyer1993deciding,beigel20053,fomin2007improved}.
For~$5$-coloring and~$6$-coloring, there are also algorithms running in time~$(2-\varepsilon)^n$ for some~$\varepsilon>0$~\cite{DBLP:conf/icalp/Zamir21}.
In contrast to $k$-SAT, for every $k>6$ the best known running time for $k$-coloring is $O^*(2^n)$, the same as computing the chromatic number. 
Thus, a fundamental open problem is whether $k$-coloring be solved exponentially faster than~$2^n$, for every $k$?

For sparse graphs, the problem was recently resolved.
\begin{theorem*}[\cite{DBLP:conf/icalp/Zamir21}]
For any~$\Delta,\alpha > 0$ there exists~$\varepsilon>0$ such that computing the chromatic number of a graph with at least~$\alpha n$ vertices of degree at most~$\Delta$ takes~$O\left(\left(2-\varepsilon\right)^n\right)$ time. Note that all sparse graphs satisfy this condition.
\end{theorem*}
As mentioned in Section~\ref{subsec:ind}, an improved algorithm for the bounded-degree case is crucial for our approach to produce results that hold for all possible degrees.

In this paper, we resolve the above open problem affirmatively for almost-regular graphs.
\begin{theorem*}[\ref{thm:maincoloring}]
    For every~$C,k$ there exists~$\varepsilon>0$ such that we can solve~$k$-coloring for graphs in which the maximum degree is at most~$C$ times the average degree in~$O\left(\left(2-\varepsilon\right)^n\right)$ time, where~$n$ is the number of vertices in the graph.
\end{theorem*}

For example, by setting~$C=\frac{1}{2\varepsilon}$ this implies faster algorithms for all dense graphs.
\begin{corollary}\label{cor:densecol}
    For every~$C,\varepsilon>0$ there exists~$\delta>0$ such that we can solve~$k$-coloring for graphs with~$n$ vertices and at least~$\varepsilon n^2$ edges in~$O\left(\left(2-\delta\right)^n\right)$ time.    
\end{corollary}

\subsection{Dense $k$-SAT}\label{subsec:ksat}
Satisfiability of Boolean formulas (usually known as SAT), is a central problem in computer science. Given a Boolean formula, the task is to decide whether there is an assignment of Boolean values to the variables of the formula under which the formula evaluates to the positive value.
If the input formula is guaranteed to be given in a standard form called $k$-CNF\footnote{The Boolean formula is commonly given in Conjunctive Normal Form (CNF), i.e., as a conjunction of disjunctions of literals. Each disjunction is called a \emph{clause}. A \emph{literal} is a variable or its negation.
A formula in which each clause contains at most $k$ literals is a $k$-CNF formula.}, then the problem of deciding whether it is satisfiable is called $k$-SAT.

Cook \cite{Cook71} and Karp \cite{karp1972reducibility} showed that $3$-SAT is NP-complete. In contrast, $2$-SAT can be solved in linear time. Many problems were shown to be NP-complete using reductions from $3$-SAT, which is thus viewed as one of the canonical NP-complete problems.

The running time of the trivial algorithm which enumerates over all possible assignments is $O^*(2^n)$. %, for some $k$.
The best known algorithms for solving $k$-SAT have exponential running times in~$n$.
%Therefore, we consider here $k$-SAT algorithms with exponential running times.
Let $c_k \in [1,2]$ be the smallest constant for which $k$-SAT can be solved in $(c_k+o(1))^n$ time.
Much effort has been put into obtaining improved upper bounds on $c_k$, especially for $3$-SAT, which has become a benchmark problem for exponential time algorithms.

A famous conjecture, called the \emph{Exponential Time Hypothesis} (ETH, see \cite{impagliazzo2001complexity}), is that $3$-SAT cannot be solved in sub-exponential time, i.e., that $c_3>1$.
A stronger conjecture that we mentioned before, known as the \emph{Strong Exponential Time Hypothesis} (SETH) is that $\lim_{k\rightarrow \infty} c_k = 2$. Both conjectures are yet to be proved or disproved and are widely used as the basis for conditional lower bounds.

The first non-trivial upper bound on~$c_k$, for any $k\ge 3$, was obtained in 1985 by Monien and Speckenmeyer \cite{monien1985solving}. They gave a deterministic algorithm showing that~$c_k<2$ for every~$k$. 
A long list of improvements followed~\cite{Rodosek96,PPZ99,Schoning02,Hertli14a} until the publication of the celebrated algorithm of Paturi, Pudlak, Saks and Zane (commonly refered to as the PPSZ algorithm) in~2001~\cite{paturi2005improved}. 
They presented a simple and elegant algorithm (with a highly non-trivial analysis that was later simplified). Their result stood as the state-of-the-art for a couple of decades, and was recently improved by~\cite{hansen2019faster} and then by~\cite{scheder2022ppsz}.

In a seminal paper, Impagliazzo, Paturi and Zane~\cite{impagliazzo2001problems} showed that the general~$k$-SAT problem can be reduced to the problem of~$k$-SAT on sparse formulas.
\begin{theorem*}[The sparsification lemma~\cite{impagliazzo2001problems}]
    For every~$k,\varepsilon>0$ any~$k$-SAT formula on~$n$ variables can be written as a disjunction of~$2^{\varepsilon n}$ $k$-SAT formulas on~$n$ variables and at most~$C_{k,\varepsilon} \cdot n$ clauses.
\end{theorem*}

Thus, every formula is at most as hard to solve as sparse formulas.
In this paper, we prove a surprising complementing statement: formulas with many clauses \emph{that are well-spread} are in fact exponentially-easier to solve than sparse (or general) formulas.
The definition of \emph{well-spread} is a bit cumbersome, so we first state the theorem informally.
\begin{theorem*}[Informal version of Theorem~\ref{thm:ksatfinal}]
    For every~$k,c$ there exists~$D_0,\varepsilon>0$ such that the following holds.
    If we have an algorithm for~$k$-SAT running in~$O(c^n)$ time, then we can solve~$k$-SAT for formulas containing a subset of at least~$Dn$ \emph{well-spread} clauses, for any~$D\geq D_0$, in~$O\left(\left(c-\varepsilon\right)^n\right)$ time. 
\end{theorem*}

The exact definition of \emph{well-spread} appears in Section~\ref{subsec:ksatalg}.
In simple words, a collections of~$Dn$ clauses is called well-spread if for some fixed-constants~$C,\varepsilon>0$ the following two conditions hold:
\begin{enumerate}
    \item Every literal appears in at most~$CD$ clauses.
    \item Every pair of literals appear together in at most~$CD^{1-\varepsilon}$ clauses.
\end{enumerate}
In Section~\ref{subsec:necess} the necessity of these conditions is extensively discussed.
In essence, the first condition prevents having many clauses that are fully contained in a negligible subset of the variables, and the second condition prevents a single clause from appearing with high multiplicity.

In random (and not sparse)~$k$-SAT formulas these conditions hold.
They also hold in every $k$-SAT formula that is very dense. If a formula contains at least~$\varepsilon n^k$ different clauses, then the average and maximum degree are both~$\Theta(n^{k-1})$, but every pair of literals can appear in at most~$\Theta(n^{k-2})$ clauses together.

\begin{corollary}\label{cor:densesat}
    For every~$c,k,\varepsilon>0$ there exists~$\delta>0$ such that the following holds.
    If we can solve~$k$-SAT on formulas with~$n$ variables in~$O(c^n)$ time, then we can solve~$k$-SAT on formulas with~$n$ variables and at least~$\varepsilon n^k$ clauses in~$O\left(\left(c-\delta\right)^n\right)$ time.    
\end{corollary}

Both Corollary~\ref{cor:densecol} and Corollary~\ref{cor:densesat} give a satisfying formulation of the natural intuition that constraint satisfaction problems should become easier if the number of constraints is very large.

\subsection{Partition Containers and Independent Sets in Almost-regular Hypergraphs}\label{subsec:parts}

One of the technical contributions of this paper is the generalization of (hyper-)graph containers to \emph{partition containers}.
For simplicity, we discuss those for graphs and not for hypergraphs of higher uniformity.

The standard graph container lemma can be phrased as follows.
\begin{theorem*}[\ref{thm:graphcontsnicealmost}]
    For every~$C > 0$ there exist~$d_0\in\mathbb{N},\;\varepsilon>0$ and a function~$r:\mathbb{N}\rightarrow [1,2]$ with~$\lim_{d\rightarrow \infty} r(d) = 1$, such that the following holds.
    Let~$G$ be a graph with~$n$ vertices, average degree~$d\geq d_0$, and maximum degree at most~$Cd$.
    There exists a collection~$\mathcal{C}$ of subsets~$C_1,C_2,\ldots,C_r \subset V(G)$ such that:
    \begin{itemize}
        \item $r \leq r(d)^n$.
        \item For every~$i$,~$|C_i| < \left(1-\varepsilon\right)n$.
        \item For every $I\in \mathcal{I}(G)$ independent sets in~$G$, there exists~$i\in [r]$ such that~$I\subseteq C_i$.
    \end{itemize}
\end{theorem*}

In particular, we are guaranteed that each individual independent set~$I\in\mathcal{I}(G)$ is fully contained in one of the containers, but we are not guaranteed anything about containing two or more independent sets.
This could be inherent: If the graph~$G$ is bipartite, then there are two independent sets whose union covers the entire graph, and thus at least two containers are needed to cover them.
We show that surprisingly, this is essentially the worst that can happen.
We modify the containers lemma and prove the following.

\begin{theorem*}[\ref{thm:partcontsalmost}]
    For every~$k,C$ there exist~$\varepsilon>0,\; d_0\in \mathbb{N}$ and a function~$r:\mathbb{N}\rightarrow [1,2]$ with~$\lim_{d\rightarrow \infty} r(d) = 1$, such that the following holds.
    Let~$G$ be a graph with~$n$ vertices, average degree~$d\geq d_0$, and maximum degree at most~$Cd$.
    There exists a collection~$\mathcal{C}$ of subsets~$C_1,C_2,\ldots,C_r \subset V(G)$ such that:
    \begin{itemize}
        \item $r \leq r(d)^n$.
        \item For every~$i$,~$|C_i| < (1-\varepsilon)n$.
        \item For every collection~$I_1,\ldots,I_k \in \mathcal{I}(G)$ of~$k$ independent sets in~$G$, there exists a partition~$[k]=A\cupdot B$ and indices~$a,b\in[r]$ such that~$\bigcup_{j\in A} I_j \subseteq C_a$ and~$\bigcup_{j\in B} I_j \subseteq C_b$.
    \end{itemize}
\end{theorem*}

Essentially, we still get a small collection of somewhat-small containers, but now we are guaranteed that every collection of~$k$ independent sets can be partitioned into two parts, such that each part is all fully contained in a single container.
This characterizes the structure of collections of independent sets in almost-regular (hyper-)graphs.

\subsection{Implications on Average-Case Hardness}
Several notions of ``hard-on-average" problems were extensively discussed in the literature (see for example the survey of Bogdanov and Trevisan on average-case complexity~\cite{bogdanov2006average}).
While in the theory of NP-completeness the difficulty of problems is measured with respect to~\emph{worst-case} instances, it is also natural to consider the~\emph{average} hardness of problems with respect to~\emph{natural} distributions of inputs.
For some natural problems and distributions, this makes a huge difference. For example, finding a Hamiltonian path in a graph is NP-complete, yet it can be detected in expected linear time in an Erd\"{o}s-Re\'{y}ni random graph~\cite{gurevich1987expected, thomason1989simple}.
On the other hand, proving that certain problems remain ``hard-on-average" is tightly related to a central open problem in cryptography.
The question of basing cryptographic primitives on complexity-theoretical assumptions, the likes of~$P\neq NP$ or~$NP\nsubseteq BPP$, dates back to Diffie and Hellman~\cite{diffie1976new}.
The security of most suggested complexity-theory-based cryptographic primitives assumes average-case hardness of NP~\cite{impagliazzo1989one}.
It remains a major open problem to base the existence of hard-on-average problems in NP on worst-case assumptions (e.g.,~$P\neq NP$).

In our work, we show that many central NP-Complete problems are \emph{inherently} easier for input instances that are \emph{somewhat-regular}. Many natural input distributions tend to be symmetrical, and hence usually produce such inputs.
This implies that for many natural problems and input distributions, the average-case complexity \emph{must} be exponentially smaller than the worst-case complexity.

\subsection{Organization of the Paper}
In Section~\ref{sec:graphconts} we present and prove the container lemma for graphs, and apply it to  the Maximum Independent Set problem, as a leading example.
In Section~\ref{sec:col} we prove our graph coloring results, using the partition containers that are constructed later in Section~\ref{sec:partcont}.
In Section~\ref{sec:sat} we oresent the general hypergraph container method and then prove our $k$-SAT related results.
Finally, in Section~\ref{sec:conc} we conclude and present open problems.

\section{Graph Containers}\label{sec:graphconts}
\subsection{Graph Container Lemma}\label{subsec:containerslemma}
In this section we present the basic approach to container lemmas.
We only discuss regular graphs, in comparison to \emph{almost}-regular \emph{hypergraphs} to which the general approach applies.
This basic approach was discovered by several researchers, most notably by Sapozhenko~\cite{sapozhenko2001number}. 
The proof we present in this section is adapted from Alon and Spencer~\cite{alon2016probabilistic} and from~\cite{balogh2018method}.

Let~$G$ be a~$d$-regular graph with vertex set~$V$ of size~$n$.
We are going to find a collection of small sets of vertices~$S\subset {V \choose {\leq qn}}$, where~$q=o_d(1)$ is a small constant, and two functions~$f:\mathcal{I}(G)\rightarrow S$ and~$g:S\rightarrow P(V)$.
The \emph{fingerprint} function~$f$ gets an independent set~$I$ in~$G$, and returns a set~$f(I)\in S$ which is a small subset of it~$f(I)\subseteq I$. This subset is called \emph{a fingerprint} of~$I$.
The \emph{container} function~$g$ gets a fingerprint~$F\in S$ and returns a larger subset of vertices~$g(F)\subseteq V$.
We are guaranteed that~$I\subseteq g(f(I))$. 
That is, the container that corresponds to a fingerprint of~$I$, must fully contain~$I$.
Crucially, the container function~$g$ depends only on the small fingerprint and not on the original independent set~$I$.
We require each container to be small.
The functions~$f,g$ are both efficiently computable and thus we can enumerate over all containers by enumerating over~${V \choose {\leq qn}}$ and applying~$g$.
This high-level approach is the same approach used in~\cite{balogh2018method} for the general hypegraph container lemma.

\begin{theorem}[\cite{sapozhenko2001number,alon2016probabilistic,balogh2018method}]\label{thm:graphconts}
    Let~$G$ be a~$d$-regular graph with vertex set~$V$ of size~$n$, and~$\varepsilon>0$.
    For~$q=\frac{1}{\varepsilon d}$ and~$S\subset {V \choose {\leq qn}}$, there exist $f:\mathcal{I}(G)\rightarrow S$ and~$g:S\rightarrow P(V)$ such that for any independent set~$I\in \mathcal{I}(G)$,~$f(I)\subseteq I \subseteq g(f(I))$, and for any~$F\in S$,~$|g(F)|\leq \left(\frac{1}{2-\varepsilon}+q\right)n$.
\end{theorem}

We begin by defining~$f$ and~$g$.
Fix an arbitrary order~$V=\{v_1,\ldots,v_n\}$ of the vertices.

Let~$I\in\mathcal{I}(G)$ be an independent set in~$G$, we define~$f(I)$ algortihmically.
We set~$F=\emptyset$ and go over the vertices of~$I$ according to the fixed order of~$V$.
When we get to a vertex~$v\in I$, we add it to~$F$ if and only if~$|N(v)\setminus N(F)|\geq \varepsilon d$. That is, only if~$v$ has at least~$\varepsilon d$ neighbors that are not yet neighbors of vertices in~$F$.
We let~$f(I)$ be the set~$F$ we end up with after going through all vertices of~$I$.

Denote by~$B(F)$ the set of all vertices~$v\in V\setminus \left(F\cup N\left(F\right)\right)$ such that~$|N(v)\cap N(F)|\geq (1-\varepsilon)d$. These are all vertices that are not in~$F$ and are not neighbors of~$F$, but at least~$(1-\varepsilon)d$ of their neighbors are also neighbors of~$F$.
For any~$F\in S$, we define~$g(F):=F\cup B(F)$.

\begin{lemma}
    For every~$I\in \mathcal{I}(G)$,~$f(I)\subseteq I$ and~$|f(I)|\leq qn$.
\end{lemma}
\begin{proof}
    By definition we only add vertices of~$I$ to~$F$ and thus~$f(I)\subseteq I$.
    Whenever we add a vertex to~$F$, the size of~$N(F)$ increases by at least~$\varepsilon d$. 
    As the size of~$N(F)$ is at most~$n$, we have~$|f(I)|\leq \frac{n}{\varepsilon d}$.
\end{proof}

\begin{lemma}
    For every~$I\in\mathcal{I}(G)$,~$I\subseteq g(f(I))$.
\end{lemma}
\begin{proof}
    By the definition of~$F$, any vertex~$v\in I\setminus f(I)$ must have at least~$(1-\varepsilon)d$ neighbors in~$N(F)$. In particular~$I\subseteq F\cup B(F)$. 
    Note that~$I$ and~$N(F)$ are disjoint as~$F\subseteq I$ and~$I$ is an independent set.
\end{proof}

\begin{lemma}\label{lem:cont_size}
    For every~$F\in S$, $|g(F)|\leq \left(\frac{1}{2-\varepsilon}+q\right)n$.
\end{lemma}
\begin{proof}
    We make two observations about the size of~$B(F)$.
    First,~$B(F)\subseteq V \setminus \left(F\cup N(F)\right)$, hence~$|B(F)|\leq n - |N(F)|$.
    Second, every~$v\in B(F)$ has at least~$(1-\varepsilon)d$ neighbors in~$N(F)$. On the other hand, each vertex in~$N(F)$ has only~$d$ neighbors (as the graph is $d$-regular). Therefore,~$|B(F)|\leq \frac{|N(F)|\cdot d}{(1-\varepsilon) d} = \frac{|N(F)|}{1-\varepsilon}$.
    We take a convex combination of these two bounds and conclude that 
    $$
    |B(F)|\leq \frac{1}{2-\varepsilon}\left(n - |N(F)|\right) + \frac{1-\varepsilon}{2-\varepsilon}\left(\frac{|N(F)|}{1-\varepsilon}\right)
    =
    \frac{n}{2-\varepsilon}
    .$$
\end{proof}

This concludes the proof of Theorem~\ref{thm:graphconts}.
\begin{remark}
    Both~$f$ and~$g$ are computable in~$O(nd)$ time.
\end{remark}

\begin{remark}\label{rmk:indsize}
    Let~$G$ be a~$d$-regular graph on~$n$ vertices.
    For every~$I\in \mathcal{I}(G)$ independent set in~$G$,~$|I|\leq \frac{n}{2}$.
\end{remark}
\begin{proof}
    There are~$d |I|$ edges adjacent to vertices in~$I$, and as~$I$ is an independent set, all of them must have an endpoint in~$V(G)\setminus I$.
    Thus,~$d|I|\leq d|V(G)\setminus I|=d\left(n-|I|\right)$ and hence~$2|I|\leq n$.
\end{proof}

We can now conclude the following formulation of the container lemma for regular graphs.
\begin{theorem}\label{thm:graphcontsnice}
    For every~$\varepsilon > 0$ there exists a function~$r:\mathbb{N}\rightarrow [1,2]$ with~$\lim_{d\rightarrow \infty} r(d) = 1$, such that the following holds.
    Let~$G$ be a~$d$-regular graph with~$n$ vertices.
    There exists a collection~$\mathcal{C}$ of subsets~$C_1,C_2,\ldots,C_r \subset V(G)$ such that:
    \begin{itemize}
        \item $r \leq r(d)^n$.
        \item For every~$i$,~$|C_i| < \left(\frac{1}{2}+\varepsilon\right)n$.
        \item For every $I\in \mathcal{I}(G)$ independent sets in~$G$, there exists~$i\in [r]$ such that~$I\subseteq C_i$.
    \end{itemize}
    Furthermore, we can compute~$\mathcal{C}$ in~$O^*\left(|\mathcal{C}|\right)$ time.
\end{theorem}
\begin{proof}
    We assume that~$\varepsilon<\frac{1}{2}$ as otherwise we may take~$\mathcal{C}=\{V\}$.
    If~$d\leq \frac{2}{\varepsilon^2}$ then we set~$r(d)=2$ and~$\mathcal{C}={V \choose \lfloor n/2 \rfloor}$.
    This satisfies the conditions by Remark~\ref{rmk:indsize}.
    Otherwise, we apply Theorem~\ref{thm:graphconts} to~$G$ with the same value of~$\varepsilon$.
    By Lemma~\ref{lem:cont_size}, for each~$F\in S$, $|F|\leq \left(\frac{1}{2-\varepsilon}+q\right)n$.
    We note that~$\frac{1}{2-\varepsilon} < \frac{1}{2} + \frac{\varepsilon}{2}$ for~$\varepsilon\in\left(0,\frac{1}{2}\right)$. 
    Thus, as~$d>\frac{2}{\varepsilon^2}$ we have~$q=\frac{1}{\varepsilon d}<\frac{\varepsilon}{2}$ and~$|F|\leq \left(\frac{1}{2}+\varepsilon\right)n$.
    We can thus set~$\mathcal{C} = Im(g)$.
    We then note that
    $$
    |\mathcal{C}| \leq {n \choose \leq qn} < 2^{H(q) n}
    ,$$
    where the last inequality is a standard bound and~$H$ is the binary entropy function (see proof in~\cite{thomas2006elements}).
    We can thus set~$r(d)=2^{H(q)}$ and finally note that
    $$
    \lim_{d\rightarrow \infty} r(d) = \lim_{d\rightarrow \infty} 2^{H\left(\frac{1}{\varepsilon d}\right)} = 1
    .$$
\end{proof}

We note that Theorem~\ref{thm:graphcontsnice} immediately implies, for example, that the number of independent sets in a~$d$-regular graph with~$n$ vertices is at most~$2^{\left(\frac{1}{2}+o_d(1)\right)n}$, as each of them is a subset of a container~\cite{sapozhenko2001number}.
Several of the most notable applications of the more general hypergraph container lemma are counting theorems that are proven in a similar fashion~\cite{saxton2015hypergraph,balogh2018method}.

We also note that the containers we get in Theorem~\ref{thm:graphconts} are also \emph{sparse}.
\begin{lemma}
    For every~$F\in S$, we have~$|E(G[g(S)])|\leq \varepsilon d n$.
\end{lemma}
\begin{proof}
    Let~$F=f(I)$ for~$I\in \mathcal{I}(G)$.
    By definition,~$g(F)=F\cup B(F)$ where~$B(F)\subseteq V\setminus N(F)$.
    Note that~$F\subseteq I$ is an independent set and thus~$g(F) \subseteq V\setminus N(F)$.
    In particular, every~$v\in F$ has no neighbors in~$g(F)$.
    Every~$v\in B(F)$ has at least~$(1-\varepsilon)d$ neighbors in~$N(F)$ and thus at most~$\varepsilon d$ neighbors in~$g(F)$.
\end{proof}

We finally state a more general variant of the graph container lemma for \emph{almost-regular} graphs.
We omit its proof, which is similar (yet slightly more involved) than the one above.
\begin{theorem}[Special case of~\cite{saxton2015hypergraph, balogh2018method}, follows from Theorem~\ref{thm:hypercontssparse}]\label{thm:graphcontsnicealmost}
    For every~$C,\varepsilon > 0$ there exist~$d_0\in\mathbb{N},\;\varepsilon'>0$ and a function~$r:\mathbb{N}\rightarrow [1,2]$ with~$\lim_{d\rightarrow \infty} r(d) = 1$, such that the following holds.
    Let~$G$ be a graph with~$n$ vertices, average degree~$d\geq d_0$, and maximum degree at most~$Cd$.
    There exists a collection~$\mathcal{C}$ of subsets~$C_1,C_2,\ldots,C_r \subset V(G)$ such that:
    \begin{itemize}
        \item $r \leq r(d)^n$.
        \item For every~$i$,~$|C_i| < \left(1-\varepsilon'\right)n$ and~$|E(G[C_i])|<\varepsilon dn$.
        \item For every $I\in \mathcal{I}(G)$ independent sets in~$G$, there exists~$i\in [r]$ such that~$I\subseteq C_i$.
    \end{itemize}
    Furthermore, we can compute~$\mathcal{C}$ in~$O^*\left(|\mathcal{C}|\right)$ time.
\end{theorem}

\subsection{Simple Applications and Maximum Independent Set}

In this section we demonstrate the power of using containers algorithmically with a simple example.

\begin{theorem}\label{thm:MISreglarge}
    Given an algorithm that solves Maximum Independent Set (MIS) in~$O(c^n)$ time, we can solve MIS in~$d$-regular graphs in~$O\left({\sqrt{c}}^{\left(1+o_d(1)\right)n}\right)$ time.
\end{theorem}
\begin{proof}
    Let~$\varepsilon'>0$.
    We apply Theorem~\ref{thm:graphcontsnice} to~$G$ with~$\varepsilon = \frac{1}{4}\varepsilon'$.
    We enumerate over all containers~$C_i$ for~$1\leq i\leq r$, and for each we use the Maximum Independent Set algorithm on~$G[C_i]$.
    We return the independent set of the largest size we found during the enumeration.
    We are guaranteed to find the maximum independent set as it is fully contained in one of the containers.
    The total running time is~$r\cdot c^{\left(1/2+\varepsilon'/4\right)n}\leq \left(r(d)c^{1/2+\varepsilon'/4}\right)^n$.
    We pick~$d_0$ such that for every~$d\geq d_0$ we have~$r(d)\leq 2^{\varepsilon'/4}$.
    Thus, for all~$d\geq d_0$ the running time is~$\sqrt{c}^{\left(1+\varepsilon'\right)n}$.
\end{proof}

If we use Theorem~\ref{thm:graphcontsnicealmost} instead of Theorem~\ref{thm:graphcontsnice} in the proof, we can also conclude the following.
\begin{theorem}\label{thm:MISalmostreglarge}
    For any~$C>1$ there exists~$\varepsilon=\varepsilon_C>0$ such the the following holds.
    Given an algorithm that solves MIS in~$O(c^n)$ time, we can solve MIS in graphs with average degree~$d$ and maximum degree bounded by~$Cd$ in~$O\left(c^{\left(1-\varepsilon+o_d(1)\right)n}\right)$ time.
\end{theorem}

We stated these two theorems in the simplest way possible, but we remark that they in fact also trivially hold for Weighted Maximum Independent Set and most other variants of the problem.

This is a simple example of how containers lead to better algorithms for (almost) regular graphs of high-degree.
On the other hand, for many problems we have faster solutions in the case of bounded-degree graphs.
Thus, our general "recipe" for faster algorithms in (almost) regular graphs is the following: If the degree is large, we use containers, otherwise, we use better algorithms for bounded-degree graphs.5

\section{Graph Coloring}\label{sec:col}
\subsection{Inclusion-Exclusion Based Graph Coloring Algorithm}\label{subsec:oldcoloringalg}
We first present a summary of Bj\"{o}rklund, Husfeldt and Koivisto's algorithm from~\cite{bjorklund2009set}.
We present a concise and partial variant of their work that applies specifically to the coloring problem, adapted from the overview in~\cite{DBLP:conf/icalp/Zamir21}. 

We begin by making the following very simple observation, yielding an equivalent phrasing of the coloring problem.
\begin{observation}
A graph $G$ is $k$-colorable if and only if its vertex set $V(G)$ can be \emph{covered} by $k$ independent sets.
\end{observation}

We need to decide whether $V(G)$ can be covered by $k$ independent sets. In order to do so, we compute the number of independent sets in every induced sub-graph and then use a simple inclusion-exclusion argument in order to compute the number of (ordered) covers of $V(G)$ by $k$ independent sets. We are interested in whether this number is positive. Complete details follow.

\begin{definition}
For a subset $V'\subseteq V(G)$ of vertices, let $i(G[V'])$ denote the number of independent sets in the induced sub-graph $G[V']$.
\end{definition}

We next show that using dynamic programming, we can quickly compute these values.
\begin{lemma}\label{computei}
We can compute the values of $i(G[V'])$ for all $V'\subseteq V$ in $O^*(2^n)$ time.
\end{lemma}
\begin{proof}
Let $v\in V'$ be an arbitrary vertex contained in $V'$.
The number of independent sets in $V'$ that do not contain $v$ is exactly $i(G[V'\setminus\{v\}])$.
On the other hand, the number of independent sets in $V'$ that do contain $v$ is exactly $i(G[V'\setminus N[v]])$.
Thus, we have
\[
i(G[V']) = i(G[V'\setminus\{v\}]) + i(G[V'\setminus N[v]]).
\]
We note that both $V'\setminus\{v\}$ and $V'\setminus N[v]$ are of size strictly less than $|V'|$. Thus, we can compute all $2^n$ values of $i(G[\cdot])$ using dynamic programming processing the sets in non-decreasing order of size.
\end{proof}

Consider the expression
\[
F(G) = \sum_{V'\subseteq V(G)} (-1)^{|V(G)|-|V'|}\cdot  i(G[V'])^k . 
\]
Using the values of $i(G[\cdot])$ computed in Lemma~\ref{computei}, we can easily compute the value of $F(G)$ by directly evaluating the above expression in $O^*(2^n)$ time.

\begin{lemma}\label{sumzero}
Let $S_1\subseteq S_2$ be sets. It holds that

\begin{equation*}
  \sum_{S_1\subseteq S \subseteq S_2} (-1)^{|S|}  =
    \begin{cases}
      0 & \text{if $S_1\neq S_2$}\\
      (-1)^{|S_2|} & \text{if $S_1=S_2$}
    \end{cases}  .     
\end{equation*}
\end{lemma}
\begin{proof}
If $S_1\subsetneq S_2$ then there exists a vertex $v\in S_2\setminus S_1$. We can pair each set $S_1\subseteq S \subseteq S_2$ with $S \triangle \{v\}$, its symmetric difference with $\{v\}$.
Clearly, in each pair of sets one is of odd size and one is of even size, and thus their signs cancel each other. Therefore, the sum is zero.
In the second case, the claim is straightforward.
\end{proof}

\begin{lemma}\label{sumzerouse}
$F(G)$ equals the number of $k$-tuples $(I_0,\ldots,I_{k-1})$ of independent sets in $G$ such that $V(G) = I_0\cup\ldots\cup I_{k-1}$.
\end{lemma}
\begin{proof}
As $i(G[V'])$ counts the number of independent sets in $G[V']$, raising it to the $k$-th power (namely, $i(G[V'])^k$) counts the number of $k$-tuples of independent sets in $G[V']$. 
%We view $F(G)$ as if it counts $k$-tuples of independent sets (in $G$) with some multiplicities.

Let $(I_0,\ldots,I_{k-1})$ be a $k$-tuple of independent sets in $G$.
It appears exactly in terms of the sum corresponding to sets $V'$ such that $I_0\cup\ldots\cup I_{k-1} \subseteq V' \subseteq V(G)$. Each time this $k$-tuple is counted, it is counted with a sign determined by the parity of $V'$. 
By Lemma~\ref{sumzero}, the sum of the signs corresponding to sets $I_0\cup\ldots\cup I_{k-1} \subseteq V' \subseteq V(G)$ is zero if $I_0\cup\ldots\cup I_{k-1} \neq V(G)$ and one if $I_0\cup\ldots\cup I_{k-1} = V(G)$.
\end{proof}

We conclude with
\begin{corollary}
$F(G)$ can be computed in time $O^*(2^n)$, and $G$ is $k$-colorable if and only if $F(G)>0$.
\end{corollary}

In a follow-up work~\cite{DBLP:conf/icalp/Zamir21}, the following was shown.
\begin{theorem}
    For every~$\alpha,\Delta>0$ there exists~$\varepsilon>0$ such that the chromatic number of graphs with at least~$\alpha n$ vertices of degree at most~$\Delta$ can be computed in~$O\left(\left(2-\varepsilon\right)^{n}\right)$ time.
\end{theorem}
\begin{corollary}\label{cor:boundeddeg}
    For every~$\Delta$ there exists~$\varepsilon>0$ such that the chromatic number of graphs with degrees bounded by~$\Delta$ can be computed in~$O\left(\left(2-\varepsilon\right)^{n}\right)$ time.
\end{corollary}

\subsection{Using Containers}\label{subsec:colwithconts}

We begin with a high level idea, to be stated precisely afterwards.
Let~$G$ be a~$d$-regular graph with~$n$ vertices, and~$k$ a fixed constant. Let~$\varepsilon>0$ be a small constant to be chosen later.
The graph container Lemma of Section~\ref{subsec:containerslemma} allows us to compute~$2^{o_d(n)}$ subsets~$C_i \subseteq V(G)$ of size~$|C_i|\leq \left(\frac{1}{2}+\varepsilon\right)n$ each, such that every independent set in~$G$ is \emph{contained} in at least one of them.
Let~$\Delta$ be a large constant to be determined later.
If~$d\leq \Delta$, then Corollary~\ref{cor:boundeddeg} gives us an improved algorithm for~$k$-coloring~$G$.
Otherwise, the number of containers is small enough for us to enumerate over~$k$-tuples of containers~$(C_1,\ldots,C_k)$ until we find one such that the~$i$-th container contains the~$i$-th color class in a~$k$-coloring of~$G$.
Thus, we may assume that we are given containers for each color class in a~$k$-coloring of~$G$. 

Precisely, we consider the following problem.
\begin{problem}[$k$-coloring given $\varepsilon$-containers]
    Given a graph~$G$ and sets~$C_1,\ldots,C_k\subseteq V(G)$ of size~$|C_i|\leq \left(\frac{1}{2}+\varepsilon\right)n$ each, decide if there is a~$k$-coloring of~$G$ in which the~$i$-th color can only be used for vertices in~$C_i$.
\end{problem}

\begin{lemma}
    If there exist~$\varepsilon,\varepsilon'>0$ such that we can solve $k$-coloring given $\varepsilon$-containers in~$O\left(\left(2-\varepsilon'\right)^n\right)$ time, then there exists~$\varepsilon''>0$ such that we can also solve~$k$-coloring for regular graphs\footnote{Without being given the containers.} in $O\left(\left(2-\varepsilon''\right)^n\right)$ time.
\end{lemma}
\begin{proof}
By enumerating over all $k$-tuples of containers, we must pass through at least one~$k$-tuple in which each container contains the respective color class, and then the algorithm that is given the containers would succeed.
The total running time of running the algorithm under the enumeration is~$\left(2^{o_d(n)}\right)^k (2-\varepsilon')^n = \left(2^{ko_d(1)}\left(2-\varepsilon'\right)\right)^n$.
In particular, there exists~$\Delta$ such that for all~$d\geq \Delta$ this running time is bounded by $\left(2-\frac{1}{2}\varepsilon'\right)^n$.
If~$d<\Delta$ then we use the algorithm of Corollary~\ref{cor:boundeddeg}.
\end{proof}

We can thus now assume we are given such containers.
Next, we go over the algorithm of Section~\ref{subsec:oldcoloringalg} and adapt it for the case of $k$-coloring given $\varepsilon$-containers.

\begin{lemma}\label{computei_containers}
We can compute the values of $i(G[V'])$ for all $V'\subseteq V$ such that~$V'\subseteq C_i$ for any~$1\leq i \leq k$ in $O^*\left(2^{(1/2+\varepsilon)n}\right)$ time.
\end{lemma}
\begin{proof}
We simply use Lemma~\ref{computei} separately for each of the graphs~$G[C_i]$ for~$1\leq i\leq k$.   
\end{proof}

Consider the expression
\[
\overline{F}(G,C_1,\ldots,C_k) = \sum_{V'\subseteq V(G)} (-1)^{|V(G)|-|V'|}\prod_{j=1}^{k} i(G[V'\cap C_j]) . 
\]

We next generalize Lemma~\ref{sumzerouse}.
\begin{lemma}\label{sumzerouse_containers}
$\overline{F}(G,C_1,\ldots,C_k)$ equals the number of $k$-tuples $(I_0,\ldots,I_{k-1})$ of independent sets in $G$ such that $V(G) = I_0\cup\ldots\cup I_{k-1}$ and~$I_j \subseteq C_j$ for every~$1\leq j\leq k$.
In particular, it is positive if and only if~$G$ is~$k$-colorable given these containers.
\end{lemma}
\begin{proof}
$\prod_{j=1}^{k} i(G[V'\cap C_j])$ counts the number of $k$-tuples $(I_0,\ldots,I_{k-1})$ of independent sets in $G[V']$, such that~$I_j \subseteq C_j$ for every~$1\leq j\leq k$.
%We view $F(G)$ as if it counts $k$-tuples of independent sets (in $G$) with some multiplicities.

Let $(I_0,\ldots,I_{k-1})$ be a $k$-tuple of independent sets in $G$ for which~$I_j \subseteq C_j$ for every~$1\leq j\leq k$.
It appears exactly in terms of the sum corresponding to sets $V'$ such that $I_0\cup\ldots\cup I_{k-1} \subseteq V' \subseteq V(G)$. Each time this $k$-tuple is counted, it is counted with a sign determined by the parity of $V'$. 
By Lemma~\ref{sumzero}, the sum of the signs corresponding to sets $I_0\cup\ldots\cup I_{k-1} \subseteq V' \subseteq V(G)$ is zero if $I_0\cup\ldots\cup I_{k-1} \neq V(G)$ and one if $I_0\cup\ldots\cup I_{k-1} = V(G)$.
\end{proof}

By Lemma~\ref{computei_containers} we can spend $O^*\left(2^{(1/2+\varepsilon)n}\right)$ time to pre-compute every value of~$i(G[\cdot])$ that appears in~$\overline{F}(G,C_1,\ldots,C_k)$.
Nevertheless, the sum still contains~$2^n$ terms and it is thus unclear if it can be computed quicker than that.

\subsection{The Extensions Sum Problem}
Let~$X$ be a set of \emph{variables}.
For a subset~$X'\subseteq X$ and a function~$f:\{0,1\}^{X'}\rightarrow \mathbb{R}$ we naturally define the \textbf{extension}~$\overline{f}:\{0,1\}^{X}\rightarrow \mathbb{R}$ as
$$
\overline{f}\left(\alpha\right) := f\left(\alpha \vert_{X'}\right)
,$$
where~$\alpha \vert_{X'}$ is the restriction of~$\alpha : X \rightarrow \{0,1\}$ to~$X'$.

\begin{definition}[The Extensions Sum Problem]
    Let~$X$ be a set and~$k$ a parameter.
    As input, we are explicitly given~$k$ subsets~$X_1,\ldots,X_k \subseteq X$ and functions~$f_i:\{0,1\}^{X_i}\rightarrow \mathbb{R}$ for~$1\leq i \leq k$.
    As output, we should compute
    $$
    \sum_{\alpha : X \rightarrow \{0,1\}} \prod_{i=1}^{k} \overline{f_i}\left(\alpha\right)
    .$$
    In the~$(k,\gamma)$-Extensions Sum Problem, we add the restriction that~$|X_i|\leq \gamma|X|$ for every~$1\leq i \leq k$.
\end{definition}
By the observations of Section~\ref{subsec:colwithconts}, we can reduce~$k$-coloring of a regular graph to a~$\left(k,\frac{1}{2}+\varepsilon\right)$-Extensions-Sum instance on~$n$ variables.
\begin{lemma}
    We can compute~$\overline{F}(G,C_1,\ldots,C_k)$ by solving an instance of~$\left(k,\frac{1}{2}+\varepsilon\right)$-Extensions-Sum. 
\end{lemma}
\begin{proof}
    First, we may assume that~$\bigcup_{j=1}^{k} C_j = V(G)$, as otherwise~$\overline{F}(G,C_1,\ldots,C_k)=0$.
    We define an instance of Extensions-Sum as follows.
    The set of variables is $X=V(G)$.
    The subsets are~$X_j = C_j$ for~$1\leq j \leq k$.
    For each~$j$, we define
    $$
    f_j(\alpha) := (-1)^{\alpha^{-1}(1)\cap \left(C_j \setminus \left(C_1\cup\ldots\cup C_{j-1}\right)\right)}\cdot i\left(G\left[\alpha^{-1}\left(1\right)\cap C_j\right]\right)
    .$$
    Let~$V'\subseteq V(G)$ and let~$\alpha : V(G)\rightarrow \{0,1\}$ be the function that indicates whether each~$v\in V(G)$ is in~$V'$.
    Then,
    $$
    \overline{f_j}(\alpha) = f_j(\alpha\vert_{X_j}) = (-1)^{V'\cap \left(C_j \setminus \left(C_1\cup\ldots\cup C_{j-1}\right)\right)}\cdot i\left(G\left[V'\cap C_j\right]\right)
    .$$
    Hence,
    $$
    \prod_{j=1}^{k} \overline{f_j}\left(\alpha\right) = 
    \prod_{j=1}^{k} \left((-1)^{V'\cap \left(C_j \setminus \left(C_1\cup\ldots\cup C_{j-1}\right)\right)}\cdot i\left(G\left[V'\cap C_j\right]\right)\right)
    = (-1)^{|V'|}\prod_{j=1}^{k} i(G[V'\cap C_j])
    .$$
\end{proof}

We thus next explore when does~$(k,\gamma)$-Extensions-Sum can be solved in~$2^{(1-\varepsilon)|X|}$ time for some~$\varepsilon>0$.
We show that~$(2,1-\varepsilon)$-Extensions-Sum can be solved in~$2^{(1-\varepsilon)|X|}$ time for any~$\varepsilon>0$.
Similarly,~$(3,1-\varepsilon)$-Extensions-Sum can be solved in~$2^{\left(1-\left(3-\omega\right)\varepsilon\right)|X|}$ time, where~$\omega<2.373$ is the matrix multiplication exponent.
On the other hand, if the hyperclique conjecture (to be described in detail later) holds, non-trivial solutions for~$k\geq 4$ are not possible. In particular, an algorithm running in time~$2^{(1-\varepsilon)|X|}$ solving~$\left(4,\frac{3}{4}\right)$ or~$\left(20,\frac{1}{2}\right)$-Extensions-Sum, for any~$\varepsilon>0$, would refute that conjecture.

\begin{lemma}\label{extsum_disjoint}
    If the subsets~$X_1,\ldots,X_k$ are disjoint, then we can solve Extensions-Sum in~$O\left(\sum_{i=1}^k 2^{|X_i|}\right)$ time.
\end{lemma}
\begin{proof}
    Denote by~$X^c := X\setminus \left(X_1\cup \ldots \cup X_k\right)$.
    The space of functions~$X \rightarrow \{0,1\}$ can be decomposed to the direct product of the spaces~$X' \rightarrow \{0,1\}$ for all~$X'=X_1,X_2,\ldots,X_k,X^c$.
    Thus, summing over all~$\alpha : X \rightarrow \{0,1\}$ is equivalent to summing over all~$\alpha \cong \left(\alpha_1,\ldots,\alpha_k,\alpha_c\right)$ where~$\alpha_i : X_i \rightarrow \{0,1\}$ and~$\alpha_c : X^c \rightarrow \{0,1\}$.
    We therefore notice that
    \begin{align*}
        \sum_{\alpha} \prod_{i=1}^{k} \overline{f_i}\left(\alpha\right)
        &=\sum_{\alpha_1,\ldots,\alpha_k,\alpha_c} \prod_{i=1}^{k} \overline{f_i}\left(\alpha\right)\\
        &=\sum_{\alpha_1,\ldots,\alpha_k,\alpha_c} \prod_{i=1}^{k} {f_i}\left(\alpha_i\right)\\
        &=\prod_{i=1}^{k} \left(\sum_{\alpha_i}  {f_i}\left(\alpha_i\right)\right) \cdot \left(\sum_{\alpha_c} 1\right)\\
        &=2^{|X^c|} \prod_{i=1}^{k} \left(\sum_{\alpha_i}  {f_i}\left(\alpha_i\right)\right).
    \end{align*}
    Thus, it is enough to separately compute the sum for each~$f_i$ in~$2^{|X_i|}$ time.
\end{proof}

\begin{lemma}\label{extsum2}
    When~$k=2$, we can solve Extensions-Sum in~$O\left(2^{|X_1|}+2^{|X_2|}\right)$ time.
\end{lemma}
\begin{proof}
    Denote by~$X_\cap := X_1\cap X_2$.
    We once again have a direct sum~$\{0,1\}^X \cong \{0,1\}^{X_\cap} \times \{0,1\}^{X\setminus X_\cap}$.
    Hence,
    $$
    \sum_{\alpha} \overline{f_1}\left(\alpha\right) \overline{f_2}\left(\alpha\right)
    =
    \sum_{\alpha_\cap : X_\cap \rightarrow \{0,1\}} \left(\sum_{\alpha' : \left(X\setminus X_\cap\right) \rightarrow \{0,1\}} \overline{f_1}\left(\alpha_\cap, \alpha'\right) \overline{f_2}\left(\alpha_\cap, \alpha'\right)\right)
    .$$
    We notice that for any fixed~$\alpha_\cap : X_\cap \rightarrow \{0,1\}$, the inner-parenthesis are also an Extensions-Sum problem; The set of variables is~$X':=X\setminus X_\cap$, the sets on which the functions are defined are~$X'_i := X_i \setminus X_\cap$, and the functions are defined as~$f'_i(\alpha'):=f_i\left(\alpha_\cap \cup \alpha'\right)$, where
\begin{equation*}
  \left(\alpha_\cap \cup \alpha'\right)(x)  :=
    \begin{cases}
      \alpha_\cap(x) & \text{if $x\in X_\cap$}\\
      \alpha'(x) & \text{otherwise}
    \end{cases}  .     
\end{equation*}
    The sets~$X'_1,X'_2$ are disjoint, and thus using Lemma~\ref{extsum_disjoint} we can compute the inner-parenthesis in~$2^{|X'_1|}+2^{|X'_2|}$ time.
    The total computation time is thus~
    $$
    2^{|X_\cap|}\cdot\left(2^{|X'_1|}+2^{|X'_2|}\right)=
    2^{|X_1|}+2^{|X_2|}
    .$$
\end{proof}

The proof of the following generalization for~$k=3$ is postponed to Appendix~\ref{appendix:es3}.
\begin{lemma}\label{extsum3}
    For any~$\varepsilon\geq 0$, we can solve~$(3,1-\varepsilon)$-Extensions-Sum in~$2^{\left(1-\left(3-\omega\right)\varepsilon\right)|X|}$ time, where~$\omega$ is the matrix multiplication exponent.
\end{lemma}

The following conjecture was posed by Lincoln, Vassilevska-Williams and Williams~\cite{lincoln2018tight}. 
In the same paper they describe several breakthroughs that would be achieved if it is false. 
It was later assumed for conditional lower bounds in several other papers (e.g., \cite{ kunnemann2020finding, williams2020truly, bringmann2021current, an2022fine}).
\begin{conjecture}
For every~$k>r\geq 3$ and~$\varepsilon>0$,~$k$-hyperclique detection in~$r$-uniform hypegraphs cannot be solved in~$O(n^{k-\varepsilon})$ time.
\end{conjecture}

In Appendix~\ref{appendix:hypercliques} we show that this conjecture implies that Extension-Sum has no non-trivial solutions for~$k>3$.
\begin{lemma}\label{extsum_lb}
    For any~$k>r\geq 2$, the problem of finding a~$k$-hyperclique in a~$r$-uniform hypergraph can be reduced to~$\left({k \choose r}, \frac{r}{k}\right)$-Extensions-Sum on~$|X|=k\lceil \log n\rceil$ variables.
\end{lemma}

\begin{corollary}
    If the hyperclique conjecture holds, for every~$k>r\geq 3$ and~$\varepsilon>0$, the~$\left({k \choose r}, \frac{r}{k}\right)$-Extensions-Sum problem cannot be solved in~$2^{(1-\varepsilon)|X|}$ time.
    This includes, in particular,~$\left(4,\frac{3}{4}\right)$ and~$\left(20,\frac{1}{2}\right)$ for the choices of~$k=4,r=3$ and~$k=6,r=3$ respectively.
\end{corollary}

\subsection{Refinements and Using Partition Containers}
While~$\left(k,\frac{1}{2}+\varepsilon\right)$-Extensions-Sum is unlikely to have a non-trivial solution for the general case, many choices of subsets~$X_1,\ldots,X_k$ make the problem much easier.

\begin{definition}
    Let~$X$ be a set. We say that a collection~$\mathcal{X}$ of subsets of~$X$ is a~$(k,\gamma)$-collection of subsets if it consists of~$k$ subsets~$X_1,\ldots,X_k$ and~$|X_i|\leq \gamma |X|$ for every~$1\leq i \leq k$.
\end{definition}

\begin{definition}\label{def:refine}
    Let~$X$ be a set and~$\mathcal{X}$ be a collection of its subsets.
    We say that~$\mathcal{X}$ has a~$(\ell, \gamma)$-refinement if there exists a partition of it to~$\ell$ parts~$\mathcal{X}=\cupdot_{i=1}^{\ell} \mathcal{X}_i$ such that for every part~$1\leq i\leq \ell$, we have~$|\cup \mathcal{X}_i| \leq \gamma |X|$.
\end{definition}

If a collection of arbitrary size has a~$(3,1-\varepsilon)$-refinement, for example, then Lemma~\ref{extsum3} results in an improved Extensions-Sum algorithm, by the following observation.

\begin{observation}\label{obs:reducerefinement}
    Let~$\mathcal{X}$ be a collection of subsets in~$X$ that has a~$(\ell,\gamma)$-refinement.
    Then, solving Extension-Sum on~$\mathcal{X}$ can be reduced to solving~$(\ell,\gamma)$-Extensions-Sum.
\end{observation}
\begin{proof}
    Let~$\mathcal{X}=\cupdot_{i=1}^{\ell} \mathcal{X}_i$ be the~$(\ell,\gamma)$-refinement of~$\mathcal{X}$.
    We define new subsets~$X_i = \cup \mathcal{X}_i$ with corresponding functions
    $$
    f'_i(\alpha) = \prod_{X'\in \mathcal{X}_i} \overline{f}_{X'}(\alpha) 
    .$$
\end{proof}

As implied by the hardness of general~$\left(k,\frac{1}{2}\right)$-Extensions-Sum then, not every~$\left(k,\frac{1}{2}\right)$-collection has such refinement.
In Appendix~\ref{appendix:partitions} we explicitly construct examples of collections without non-trivial refinements.
On the other hand, in the main technical part of this paper we show that the container lemma can be generalized in a way that produces only collections of containers that have~$(2,1-\varepsilon)$-refinements. The proof of the following Theorem appears in Section~\ref{sec:partcont}.

\begin{theorem*}[\ref{thm:partconts}]
    For every~$k$ there exist~$\varepsilon>0,\; d_0\in \mathbb{N}$ and a function~$r:\mathbb{N}\rightarrow [1,2]$ with~$\lim_{d\rightarrow \infty} r(d) = 1$, such that the following holds.
    Let~$G$ be a~$d$-regular graph with~$n$ vertices and~$d\geq d_0$.
    There exists a collection~$\mathcal{C}$ of subsets~$C_1,C_2,\ldots,C_r \subset V(G)$ such that:
    \begin{itemize}
        \item $r \leq r(d)^n$.
        \item For every~$i$,~$|C_i| < (1-\varepsilon)n$.
        \item For every collection~$I_1,\ldots,I_k \in \mathcal{I}(G)$ of~$k$ independent sets in~$G$, there exists a partition~$[k]=A\cupdot B$ and indices~$a,b\in[r]$ such that~$\bigcup_{j\in A} I_j \subseteq C_a$ and~$\bigcup_{j\in B} I_j \subseteq C_b$.
    \end{itemize}
    Furthermore, we can compute~$\mathcal{C}$ in~$O^*\left(|\mathcal{C}|\right)$ time.
\end{theorem*}

We are finally ready to prove our main results.
\begin{theorem}
    For every~$k$ there exists~$\varepsilon>0$ such that we can solve~$k$-coloring for regular graphs in~$O\left(\left(2-\varepsilon\right)^n\right)$ time, where~$n$ is the number of vertices in the graph.
\end{theorem}
\begin{proof}
Let~$\Delta=\Delta_k\geq d_0$ be a constant to be chosen later.
Consider an input~$d$-regular graph~$G$ with~$n$ vertices.
If~$d<\Delta$, we run the algorithm of Corollary~\ref{cor:boundeddeg}.
Otherwise, generate the~$r\leq r(d)^n$ partition containers~$C_1,\ldots,C_r$ of~$G$ using Theorem~\ref{thm:partconts}.
Enumerate over all pairs~$C_i,C_j$ of containers, and over all~$2^k$ partitions~$[k]=A\cupdot B$.
For each such choice, we compute~$\overline{F}(G,C'_1,\ldots,C'_k)$ using the algorithm of Section~\ref{subsec:colwithconts}, where~$C'_\ell = C_i$ if~$\ell\in A$ and~$C'_\ell=C_j$ if~$\ell\in B$.
We solve the Extensions-Sum problem in the computation of~$\overline{F}$ using Observation~\ref{obs:reducerefinement} and Lemma~\ref{extsum2}.

Given~$C_i,C_j,A,B$, the running time is~$O\left(2^{(1-\varepsilon)n}\right)$ where~$(1-\varepsilon)n$ is the bound on the size of each container.
The number of choices for~$C_i,C_j,A,B$ is~$2^k r(d)^{2n}$.
We pick~$\Delta\geq d_0$ to be large enough so that~$r(d)^{2} 2^{(1-\varepsilon)}$ is bounded away from~$2$ for all~$d\geq \Delta$.
\end{proof}

To generalize the result from regular graphs to almost-regular graphs, we simply need to replace the use of Theorem~\ref{thm:partconts} with the following generalization.
\begin{theorem*}[\ref{thm:partcontsalmost}]
    For every~$k,C$ there exist~$\varepsilon>0,\; d_0\in \mathbb{N}$ and a function~$r:\mathbb{N}\rightarrow [1,2]$ with~$\lim_{d\rightarrow \infty} r(d) = 1$, such that the following holds.
    Let~$G$ be a graph with~$n$ vertices, average degree~$d\geq d_0$, and maximum degree at most~$Cd$.
    There exists a collection~$\mathcal{C}$ of subsets~$C_1,C_2,\ldots,C_r \subset V(G)$ such that:
    \begin{itemize}
        \item $r \leq r(d)^n$.
        \item For every~$i$,~$|C_i| < (1-\varepsilon)n$.
        \item For every collection~$I_1,\ldots,I_k \in \mathcal{I}(G)$ of~$k$ independent sets in~$G$, there exists a partition~$[k]=A\cupdot B$ and indices~$a,b\in[r]$ such that~$\bigcup_{j\in A} I_j \subseteq C_a$ and~$\bigcup_{j\in B} I_j \subseteq C_b$.
    \end{itemize}
    Furthermore, we can compute~$\mathcal{C}$ in~$O^*\left(|\mathcal{C}|\right)$ time.
\end{theorem*}

We conclude the following.
\begin{theorem}\label{thm:maincoloring}
    For every~$C,k$ there exists~$\varepsilon>0$ such that we can solve~$k$-coloring for graphs in which the maximum degree is at most~$C$ times the average degree in~$O\left(\left(2-\varepsilon\right)^n\right)$ time, where~$n$ is the number of vertices in the graph.
\end{theorem}
\begin{proof}
Let~$\Delta=\Delta_k\geq d_0$ be a constant to be chosen later.
Consider an input graph~$G$ with~$n$ vertices, average degree~$d$, and maximum degree bounded by~$Cd$.
If~$d<\Delta$, we run the algorithm of Corollary~\ref{cor:boundeddeg} as the degrees are bounded by~$C\Delta$.
Otherwise, generate the~$r\leq r(d)^n$ partition containers~$C_1,\ldots,C_r$ of~$G$ using Theorem~\ref{thm:partcontsalmost}.
Enumerate over all pairs~$C_i,C_j$ of containers, and over all~$2^k$ partitions~$[k]=A\cupdot B$.
For each such choice, we compute~$\overline{F}(G,C'_1,\ldots,C'_k)$ using the algorithm of Section~\ref{subsec:colwithconts}, where~$C'_\ell = C_i$ if~$\ell\in A$ and~$C'_\ell=C_j$ if~$\ell\in B$.
We solve the Extensions-Sum problem in the computation of~$\overline{F}$ using Observation~\ref{obs:reducerefinement} and Lemma~\ref{extsum2}.

Given~$C_i,C_j,A,B$, the running time is~$O\left(2^{(1-\varepsilon)n}\right)$ where~$(1-\varepsilon)n$ is the bound on the size of each container.
The number of choices for~$C_i,C_j,A,B$ is~$2^k r(d)^{2n}$.
We pick~$\Delta\geq d_0$ to be large enough so that~$r(d)^{2} 2^{(1-\varepsilon)}$ is bounded away from~$2$ for all~$d\geq \Delta$.
\end{proof}

\section{Graph Partition Containers}\label{sec:partcont}
\subsection{Regular Graphs}
In this section we generalize the graph container lemma and prove the following.
\begin{theorem}\label{thm:partconts}
    For every~$k$ there exist~$\varepsilon>0,\; d_0\in \mathbb{N}$ and a function~$r:\mathbb{N}\rightarrow [1,2]$ with~$\lim_{d\rightarrow \infty} r(d) = 1$, such that the following holds.
    Let~$G$ be a~$d$-regular graph with~$n$ vertices and~$d\geq d_0$.
    There exists a collection~$\mathcal{C}$ of subsets~$C_1,C_2,\ldots,C_r \subset V(G)$ such that:
    \begin{itemize}
        \item $r \leq r(d)^n$.
        \item For every~$i$,~$|C_i| < (1-\varepsilon)n$.
        \item For every collection~$I_1,\ldots,I_k \in \mathcal{I}(G)$ of~$k$ independent sets in~$G$, there exists a partition~$[k]=A\cupdot B$ and indices~$a,b\in[r]$ such that~$\bigcup_{j\in A} I_j \subseteq C_a$ and~$\bigcup_{j\in B} I_j \subseteq C_b$.
    \end{itemize}
    Furthermore, we can compute~$\mathcal{C}$ in~$O^*\left(|\mathcal{C}|\right)$ time.
\end{theorem}

While the standard container lemma explores the structure of (single) independent sets in a graph, this version explores the structure of small collections of independent sets at the same time.

The proof of this theorem is fully contained in this paper.
In this section we build upon the terminology and proofs of Section~\ref{subsec:containerslemma}.
Let~$G$ be a~$d$-regular graph with~$n$ vertices, and let~$I_1, \ldots, I_k$ be independent sets in~$G$.
We prove that their corresponding containers must have a bounded refinement with two parts.

\begin{lemma}\label{lem:contrefine}
    Let~$G$ be a~$d$-regular graph with vertex set~$V$ of size~$n$, and~$\varepsilon>0$.
    Consider~$q,S,f,g$ given by Theorem~\ref{thm:graphconts}.
    For any~$F_1,\ldots,F_k \in S$, the collection~$\{g(F_1),\ldots,g(F_k)\}$ of subsets in~$V$ has a~$(2,\gamma)$-refinement\footnote{See Definition~\ref{def:refine}} for~$\gamma\leq \frac{1-2^{-k}}{1-\varepsilon} + kq$.
\end{lemma}

We begin by proving a few useful lemmas.

\begin{lemma}\label{lem:partcap}
    For any~$F_1,\ldots,F_r \subseteq V$,~$|\bigcup_{i=1}^{r} B(F_i)| \leq n - |\bigcap_{i=1}^{r} N(F_i)|$.
\end{lemma}
\begin{proof}
    By definition, for each~$i\in[r]$ we have~$B(F_i) \subseteq V\setminus \left(F_i \cup N(F_i)\right) \subseteq V \setminus \bigcap_{i=1}^{r} N(F_i)$, where the second inequality follows as~$\bigcap_{i=1}^{r} N(F_i) \subseteq N(F_i) \subseteq F_i \cup N(F_i)$ for every~$i\in [r]$.
\end{proof}

\begin{lemma}\label{lem:partcup}
    For any~$F_1,\ldots,F_r \subseteq V$,~$|\bigcup_{i=1}^{r} B(F_i)| \leq \frac{1}{1-\varepsilon}|\bigcup_{i=1}^{k} N(F_i)|$.
\end{lemma}
\begin{proof}
    If~$v\in \bigcup_{i=1}^{r} B(F_i)$, then exists~$i\in [r]$ such that~$v$ has at least~$(1-\varepsilon)d$ neighbors in~$N(F_i)$, and in particular at least~$(1-\varepsilon)d$ neighbors in~$\bigcup_{i=1}^{k} N(F_i)$.
    On the other hand, as~$G$ is~$d$-regular, the set~$\bigcup_{i=1}^{k} N(F_i)$ has at most~$d|\bigcup_{i=1}^{k} N(F_i)|$ adjacent edges.
    Hence,~$|\bigcup_{i=1}^{r} B(F_i)| \leq \frac{d|\bigcup_{i=1}^{k} N(F_i)|}{(1-\varepsilon)d}$.
\end{proof}

\begin{lemma}\label{lem:intuniparts}
    Let~$V$ be a set of size~$n:=|V|$ and let~$V_1,\ldots,V_k\subseteq V$ be~$k$ subsets of it.
    There exists a partition~$[k]=A\cupdot B$ such that~$|\bigcap_{i \in A} V_i| \geq \frac{n}{2^k}$ and~$|\bigcup_{i \in B} V_i| \leq \left(1-\frac{1}{2^k}\right) n$.
    We define the intersection of zero subsets to be~$V$ and the union of zero subsets to be the empty set~$\emptyset$.
\end{lemma}
\begin{proof}
    Let~$Venn : V\rightarrow \{0,1\}^{[k]}$ map each~$v\in V$ to a binary vector of length~$k$ in which the~$i$-th entry is~$1$ if and only if~$v\in V_i$.
    By the pigeon-hole principle, there exists a vector~$\textbf{w} \in \{0,1\}^{[k]}$ such that~$|Venn^{-1}(\textbf{w})|\geq \frac{n}{2^k}$.
    We define a partition of~$[k]$ by~$A=\{i\in [k]\;\mid\; w_i = 1\}$ and~$B=\{i\in [k]\;\mid\; w_i = 0\}$.
    By definition,~$Venn^{-1}(\textbf{w})\subseteq \bigcap_{i \in A} V_i$.
    This follows as if~$v\in Venn^{-1}(\textbf{w})$ and~$i \in A$ then~$v\in V_i$.
    Furthermore,~$\bigcup_{i \in B} V_i \subseteq V\setminus Venn^{-1}(\textbf{w})$.
    This follos as if~$v\in \bigcup_{i \in B} V_i$ then there is some~$i\in[k]$ such that~$v\in V_i$ but~$w_i=0$ and thus~$v\notin Venn^{-1}(\textbf{w})$.
\end{proof}

\begin{proof}[Proof of Lemma~\ref{lem:contrefine}]
    Apply Lemma~\ref{lem:intuniparts} to the sets~$N(F_1),N(F_2),\ldots,N(F_k)\subseteq V$ to get a partition~$[k]=A\cupdot B$ satisfying the Lemma's statement.
    We prove that this partition is also a good refinement for~$g(F_1),\ldots,g(F_k)$.
    By Lemma~\ref{lem:partcap} we have~$|\bigcup_{i\in A} B(F_i)| \leq n - |\bigcap_{i\in A} N(F_i)| \leq \left(1-2^{-k}\right)n$.
    By Lemma~\ref{lem:partcup} we have~$|\bigcup_{i\in B} B(F_i)| \leq \frac{1}{1-\varepsilon}|\bigcup_{i\in B} N(F_i)| \leq \frac{1}{1-\varepsilon} \left(1-2^{-k}\right)n$.
    We finish by noting that~$$|\bigcup_{i\in A} g(F_i)| = |\bigcup_{i\in A} \left(F_i \cup B(F_i)\right)| \leq |\bigcup_{i\in A} F_i| + |\bigcup_{i\in A}  B(F_i)| \leq |A|\cdot qn + |\bigcup_{i\in A}  B(F_i)|.$$
    Similarly,~$|\bigcup_{i\in B} g(F_i)|  \leq |B|\cdot qn + |\bigcup_{i\in B}  B(F_i)|.$
\end{proof}

We are now ready to prove Theorem~\ref{thm:partconts}.
\begin{proof}[Proof of Theorem~\ref{thm:partconts}]
    We set~$\varepsilon = 2^{-(k+2)},\; d_0 =k2^{2k+3}$.
    We apply Theorem~\ref{thm:graphcontsnice} to~$G$ with~$\varepsilon' = 2^{-(k+1)}$, denote by~$r',\mathcal{C}'$ the function and collection of containers it produces.
    We define the set of partition containers to be~$\mathcal{C} = \{\bigcup_{C\in \mathcal{C}''}C \;\mid\; \mathcal{C}'' \in {\mathcal{C}' \choose \leq k} \; \wedge\; |\bigcup_{C\in \mathcal{C}''}C|\leq (1-\varepsilon)n\}$. That is, the partition containers are the unions of any collection of~$\leq k$ (standard) containers in~$\mathcal{C}'$ that is of size at most~$(1-\varepsilon)n$.
    We can thus bound the number of partition containers by~$r(d)=\min\{r'(d)^k,2\}$.
    We still have~$\lim_{d\rightarrow \infty} r(d) = 1$.
    By Lemma~\ref{lem:contrefine}, for every collection of~$k$ independent sets~$I_1,\ldots,I_k\in \mathcal{I}(G)$, their corresponding containers in~$\mathcal{C'}$ have a~$(2,\gamma)$-refinement with~$\gamma \leq \frac{1-2^{-k}}{1-\varepsilon'} + kq$.
    As~$d\geq d_0$ we have
    \begin{align*}
    \frac{1-2^{-k}}{1-\varepsilon'} + kq
    &=
    \frac{1-2^{-k}}{1-2^{-(k+1)}} + \frac{k}{2^{-(k+1)} d}
    \\&=
    \frac{1-2^{-(k+1)}-2^{-(k+1)}}{1-2^{-(k+1)}} + \frac{k2^{k+1}}{d}
    \\&=
    1 - \frac{2^{-(k+1)}}{1-2^{-(k+1)}} + \frac{k2^{k+1}}{d}
    \\&\leq
    1 - 2^{-(k+1)} + \frac{k2^{k+1}}{k2^{2k+3}}
    \\&\leq
    1 - 2^{-(k+2)} = 1-\varepsilon
    ,\end{align*}
    and thus the parts of the refinement appear in~$\mathcal{C}$.
\end{proof}

\subsection{Almost-regular Graphs}
In this section we use Theorem~\ref{thm:graphcontsnicealmost} to prove the following generalization of Theorem~\ref{thm:partconts}.
\begin{theorem}\label{thm:partcontsalmost}
    For every~$k,C$ there exist~$\varepsilon>0,\; d_0\in \mathbb{N}$ and a function~$r:\mathbb{N}\rightarrow [1,2]$ with~$\lim_{d\rightarrow \infty} r(d) = 1$, such that the following holds.
    Let~$G$ be a graph with~$n$ vertices, average degree~$d\geq d_0$, and maximum degree at most~$Cd$.
    There exists a collection~$\mathcal{C}$ of subsets~$C_1,C_2,\ldots,C_r \subset V(G)$ such that:
    \begin{itemize}
        \item $r \leq r(d)^n$.
        \item For every~$i$,~$|C_i| < (1-\varepsilon)n$.
        \item For every collection~$I_1,\ldots,I_k \in \mathcal{I}(G)$ of~$k$ independent sets in~$G$, there exists a partition~$[k]=A\cupdot B$ and indices~$a,b\in[r]$ such that~$\bigcup_{j\in A} I_j \subseteq C_a$ and~$\bigcup_{j\in B} I_j \subseteq C_b$.
    \end{itemize}
    Furthermore, we can compute~$\mathcal{C}$ in~$O^*\left(|\mathcal{C}|\right)$ time.
\end{theorem}

Our proof is similar to the one for the regular case, but we now need to use the sparsity of each container.

\begin{lemma}\label{lem:refinematch}
    Let~$V$ be a set and~$V_1,\ldots,V_k \subseteq V$ subsets of it.
    Let~$M\subseteq {V \choose 2}$ be a collection of \textbf{pairwise-disjoint} pairs of items in~$V$.
    Assume that for every~$i\in[k]$ and every~$e\in M$,~$e\nsubseteq V_i$. That is, no pair in~$M$ is fully contained in any~$V_i$.
    Then,~$\{V_1,\ldots,V_k\}$ has a refinement to two parts of size at most~$|V|-2^{-k}|M|$.
\end{lemma}
\begin{proof}
    We pick an arbitrary order~$e=(e_0,e_1)$ to the to items in every pair~$e\in M$.
    By the assumption, for every~$e\in M$ and~$i\in[k]$, either~$e_0\notin V_i$ or~$e_1\notin V_i$.
    
    Let~$\xi : M\rightarrow \{0,1\}^{[k]}$ map each~$e\in M$ to a binary vector of length~$k$ in which the~$i$-th entry is~$0$ if and only if~$e_0\notin V_i$.
    By the pigeon-hole principle, there exists a vector~$\textbf{w} \in \{0,1\}^{[k]}$ such that~$|\xi^{-1}(\textbf{w})|\geq \frac{|M|}{2^k}$.
    We define a partition of~$[k]$ by~$A=\{i\in [k]\;\mid\; w_i = 0\}$ and~$B=\{i\in [k]\;\mid\; w_i = 1\}$.

    Denote by~$W_0:=\{e_0\;|\;e\in\xi^{-1}(\textbf{w})\}$ and by~$W_1:=\{e_1\;|\;e\in\xi^{-1}(\textbf{w})\}$.
    Let~$i\in A$, we show that~$V_i\subseteq V\setminus W_0$.
    This holds as~$w_i=0$ and thus for any~$e\in\xi^{-1}(\textbf{w})$ we have~$e_0\notin V_i$.
    Symmetrically, for every~$i\in B$ we have~$V_i\subseteq V\setminus W_1$.
    In particular,~$\bigcup_{i\in A}V_i \subseteq V\setminus W_0$ and~$\bigcup_{i\in B}V_i \subseteq V\setminus W_1$.

    Since the pairs in~$E'$ are pairwise disjoint,~$|W_0|=|W_1|=|\xi^{-1}(\textbf{w})|\geq \frac{|M|}{2^k}$.
\end{proof}

\begin{lemma}\label{lem:largematch}
    Let~$G$ be a graph with maximum degree at most~$\Delta$, and~$E'\subseteq E(G)$ be a subset of its edges.
    There exists a matching~$M\subseteq E'$ of size~$|M|\geq \frac{|E'|}{\Delta+1}$.
\end{lemma}
\begin{proof}
    Vizing's theorem (see~\cite{alon2016probabilistic}) states that any graph with maximum degree~$\Delta$ is~$(\Delta+1)$ edge-colorable\footnote{Furthermore, such a coloring can be found in polynomial time.}.
    In particular, there exists a partition~$E(G)=M_1\cupdot\ldots\cupdot M_{\Delta+1}$ such that every~$M_i$ is a matching.
    We pick~$i$ such that~$|E'\cap M_i|$ is maximal and set~$M=E'\cap M_i$.
\end{proof}

\begin{lemma}\label{lem:sparsetorefine}
    Let~$G$ be a graph with~$n$ vertices, average degree~$d$, and maximum degree at most~$Cd$.
    Suppose that the subsets~$C_1,\ldots,C_k\subseteq V(G)$ have~$|E(G[C_i])|<\varepsilon dn$ for every~$i\in[k]$.
    Then,~$\{C_1,\ldots,C_k\}$ has a~$\left(2,1-\frac{\frac{1}{2}-k\varepsilon}{C2^k}\right)$-refinement.
\end{lemma}
\begin{proof}
Let~$E':=E(G)\setminus\bigcup_{i=1}^{k} E(G[C_i])$ be the set of edges in~$G$ that are not contained in any~$C_i$.
By the assumptions,~$|E'|>\frac{dn}{2} - k\cdot \varepsilon dn = \left(\frac{1}{2}-k\varepsilon\right) dn$.
By Lemma~\ref{lem:largematch}, there exists a matching~$M\subseteq E'$ of size~$|M|\geq \frac{\left(\frac{1}{2}-k\varepsilon\right) dn}{Cd} = \frac{\frac{1}{2}-k\varepsilon}{C} n$.
By the definition of~$E'$, no edge of~$M$ is fully contained in any~$C_i$. Thus, by Lemma~\ref{lem:refinematch} we have a refinement of~$\{C_1,\ldots,C_k\}$ to two parts of size at most 
$$
|V(G)| - 2^{-k} |M| = n - 2^{-k}\cdot \frac{\frac{1}{2}-k\varepsilon}{C} n
.$$
\end{proof}

Using Lemma~\ref{lem:sparsetorefine}, we can prove Theorem~\ref{thm:partcontsalmost} in the same manner in which we proved Theorem~\ref{thm:partconts}, by replacing the use of Theorem~\ref{thm:graphcontsnice} with Theorem~\ref{thm:graphcontsnicealmost} and setting the sparsity of each container to be~$\varepsilon<\frac{1}{2k}$.

\begin{proof}[Proof of Theorem~\ref{thm:partcontsalmost}]
    We apply Theorem~\ref{thm:graphcontsnicealmost} to~$G$ with parameters~$C$ and~$\varepsilon = \frac{1}{4k}$, denote by~$r',\mathcal{C}'$ the function and collection of containers it produces.
    Denote by~$\varepsilon'' = \frac{\frac{1}{2}-k\varepsilon}{C2^k} = \frac{1}{C2^{k+2}}$.
    We define the set of partition containers to be~$\mathcal{C} = \{\bigcup_{C\in \mathcal{C}''}C \;\mid\; \mathcal{C}'' \in {\mathcal{C}' \choose \leq k} \; \wedge\; |\bigcup_{C\in \mathcal{C}''}C|\leq (1-\varepsilon'')n\}$. That is, the partition containers are the unions of any collection of~$\leq k$ (standard) containers in~$\mathcal{C}'$ that is of size at most~$(1-\varepsilon'')n$.
    We can thus bound the number of partition containers by~$r(d)=\min\{r'(d)^k,2\}$.
    We still have~$\lim_{d\rightarrow \infty} r(d) = 1$.
    By Lemma~\ref{lem:sparsetorefine}, for every collection of~$k$ independent sets~$I_1,\ldots,I_k\in \mathcal{I}(G)$, their corresponding containers in~$\mathcal{C'}$ have a~$(2,1-\varepsilon'')$-refinement and thus the parts of the refinement appear in~$\mathcal{C}$.
\end{proof}

We note that the proof in this section generalizes to hypergraphs of arbitrary uniformity. In this paper we only use the graph version of the theorem and thus the more general proof is omitted from this version of the paper.

\section{Hypergraph Containers and~$k$-SAT}\label{sec:sat}
\subsection{Review of Hypergraph Container Lemmas}\label{subsec:hyperconts}
In this section we present the generalization of Theorems~\ref{thm:graphconts} and~\ref{thm:graphcontsnicealmost} to hypergraphs. 
We use the terminology and version of~\cite{balogh2015independent}.

Let~$\mathcal{H}$ be a~$r$-uniform hypergraph.
That is, every edge~$e\in E(\mathcal{H})$ is a set of exactly~$r$ vertices of~$V(\mathcal{H})$.

\begin{definition}[co-degrees]
    For a subset of vertices~$T\subseteq V(\mathcal{H})$, we define the co-degree of~$T$ to be~$\deg_\mathcal{H}(T) := |\{e\in E(\mathcal{H})\;\mid\; T\subseteq e \}|$.
\end{definition}

\begin{definition}[max co-degrees]
    For any~$1\leq i \leq r$, we define the~$i$-th max co-degree in~$\mathcal{H}$ to be~$\Delta_i(\mathcal{H}):=\max \{ \deg_\mathcal{H}(T)\;\mid\; T\subseteq V(\mathcal{H}),\; |T|=i\}$.
\end{definition}

\begin{theorem}[Proposition 3.1 in \cite{balogh2015independent}]\label{thm:hyperconts}
    For every~$r\in \mathbb{N}$ and all positive~$C$, there exists a positive constant~$\delta$ such that the following holds.
    Let~$p\in(0,1)$ and suppose~$\mathcal{H}$ is a~$r$-uniform hypergraph such that, for every~$1\leq i\leq r$,
    $$\Delta_i(\mathcal{H}) \leq C \cdot p^{i-1} \frac{|E(\mathcal{H})|}{|V(\mathcal{H})|}.$$
    Then there exists a family~$S\subseteq {V(\mathcal{H}) \choose {\leq (r-1)p|V(\mathcal{H})|}}$ and functions~$f:\mathcal{I}(\mathcal{H})\rightarrow S$ and~$g:S\rightarrow P(V(\mathcal{H}))$ such that for every~$I\in\mathcal{I}(\mathcal{H})$ we have~$f(I)\subseteq I\subseteq g(f(I))$ and~$|g(f(I))|\leq \left(1-\delta+\left(r-1\right)p\right)|V(\mathcal{H})|$.
\end{theorem}

We note that Theorem~\ref{thm:hyperconts} implies Theorem~\ref{thm:graphcontsnicealmost} by setting~$r=2,\;p=\frac{|V(G)|}{C|E(G)|}=\frac{2}{Cd}$.
Generally, to get a number of containers that decreases with the degree we need to choose~$p$ as a function that vanishes as the average degree grows.
We remark that in a similar fashion to Section~\ref{subsec:containerslemma}, the functions~$f,g$ can be computed efficiently.

In many applications of the hypergraph container method, it is required that each container does not fully contain many edges of~$\mathcal{H}$.
This is a corollary of Theorem~\ref{thm:hyperconts}, as if a container~$C_i\in g(S)$ has~$\geq \varepsilon |E(\mathcal{H})|$ edges, then the average degrees in the induced sub-hypergraph~$\mathcal{H}[C_i]$ are similar to those in~$\mathcal{H}$ and we can iteratively use Theorem~\ref{thm:hyperconts}.
We follow with the rigorous statement.

\begin{theorem}[Special case of Theorem 2.2 in \cite{balogh2015independent}]\label{thm:hypercontssparse}
    For every~$r\in \mathbb{N}$ and all positive~$C$ and~$\varepsilon$, there exists a positive constant~$M$ such that the following holds.
    Let~$p\in(0,1)$ and suppose~$\mathcal{H}$ is a~$r$-uniform hypergraph such that, for every~$1\leq i\leq r$,
    $$\Delta_i(\mathcal{H}) \leq C \cdot p^{i-1} \frac{|E(\mathcal{H})|}{|V(\mathcal{H})|}.$$
    Then there exists a family~$S\subseteq {V(\mathcal{H}) \choose {\leq Mp|V(\mathcal{H})|}}$ and functions~$f:\mathcal{I}(\mathcal{H})\rightarrow S$ and~$g:S\rightarrow P(V(\mathcal{H}))$ such that for every~$I\in\mathcal{I}(\mathcal{H})$ we have~$f(I)\subseteq I\subseteq g(f(I))$ and~$|E(\mathcal{H}[g(f(I))])|< \varepsilon|E(\mathcal{H})|$.
\end{theorem}

Theorem~\ref{thm:hypercontssparse} is in fact also useful to get better bounds on the size of containers.
For example, Theorem~\ref{thm:hypercontssparse} implies Theorem~\ref{thm:graphcontsnice}, which is not implied directly by Theorem~\ref{thm:hyperconts}.
This follows because in a~$d$-regular graph every induced graph on at least~$\left(\frac{1}{2}+\delta\right)n$ vertices must contain at least~$\left(\frac{1}{2}+\delta\right)nd - \left(\frac{1}{2}-\delta\right)nd = 2\delta nd$ edges. Thus, every induced subgraph that contains at most~$\varepsilon |E(G)|$ edges can contain at most~$\left(\frac{1}{2} + \frac{1}{4}\varepsilon\right)n$ vertices.
This is a simple example of a concept called \emph{supersaturation} that is frequently used together with the container lemma.

\subsection{Dense $k$-SAT Algorithm}\label{subsec:ksatalg}
In this section we give a better algorithm for~$k$-SAT on formulas that have many clauses that are \emph{well-spread}. We will go deeply into the exact definition of this term and for its necessity.

The first observation we make is that~$k$-SAT can be reduced to finding independent sets in an appropriate~$k$-uniform hypergraph.

\begin{definition}
    Let~$\varphi$ be a~$k$-SAT formula on the variables~$X=\{x_1,\ldots,x_n\}$.
    We denote by~$\mathcal{H}_\varphi$ the~$k$-uniform hypergraph defined as follows.
    $V(\mathcal{H}_\varphi):=\{x_1,\overline{x_1},x_2,\overline{x_2},\ldots,x_n,\overline{x_n}\}$ is the set of all \emph{literals} that can appear in~$\varphi$.
    We have an edge~$e\in E(\mathcal{H}_\varphi)$ for every clause~$C$ in~$\varphi$.
    For a clause~$C=(\ell_1\vee \ell_2 \vee \ldots \vee \ell_k)$, where each~$\ell_i$ is a literal, we add the edge~$e=\{\overline{\ell_1},\ldots,\overline{\ell_k}\}$.
\end{definition}

We note that the co-degrees in~$\mathcal{H}_\varphi$ have natural descriptions in terms of~$\varphi$, the max co-degree~$\Delta_i(\mathcal{H}_\varphi)$ is the maximum number of clauses in~$\varphi$ that all intersect in at least~$i$ literals.

We first show that each satisfying assignment of~$\varphi$ corresponds to an independent set of~$\mathcal{H}_\varphi$. 
\begin{lemma}\label{lem:satind}
    Let~$\alpha$ be a satisfying assignment of~$\varphi$.
    We denote by~$I_\alpha$ the set of size~$n$ that contains, for each~$i$,~$x_i$ if~$\alpha(x_i)=1$ and~$\overline{x_i}$ otherwise.
    Then,~$I_\alpha\in\mathcal{I}(\mathcal{H}_\varphi)$.
\end{lemma}
\begin{proof}
    Let~$e\in E(\mathcal{H}_\varphi)$ be an edge corresponding to a clause~$C=(\ell_1\vee \ell_2 \vee \ldots \vee \ell_k)$.
    Since~$\alpha$ satisfies~$C$, there is some~$i$ such that~$\alpha(\ell_i)=1$.
    Hence,~$\overline{\ell_i}\notin I_\alpha$ yet~$\overline{\ell_i}\in e$.
\end{proof}

We next show that any subgraph of~$\mathcal{H}_\varphi$ corresponds to a smaller~$k$-SAT formula in a meaningful way.
\begin{definition}
    For any~$V'\subseteq V(\mathcal{H}_\varphi)$, we denote by~$\varphi[V']$ the formula~$\varphi$ after partially assigning the following values to some of its variables.
    For every variable~$x_i$, if~$x_i\notin V'$, we assign~$x_i\leftarrow 0$; If~$\overline{x_i}\notin V'$ we assign~$x_i\leftarrow 1$; If both~$x_i,\overline{x_i}\notin V'$ then the formula~$\varphi[V']$ is a contradiction.
\end{definition}

\begin{lemma}\label{lem:subsat}
    For any~$V'\subseteq V(\mathcal{H}_\varphi)$ if~$\varphi[V']$ is satisfiable then~$\varphi$ is also satisfiable.
    Moreover, if~$\alpha$ is a satisfying assignment of~$\varphi$ and~$I_\alpha \subseteq V' \subseteq  V(\mathcal{H}_\varphi)$, then~$\varphi[V']$ is satisfiable.
\end{lemma}
\begin{proof}
    The first part follows as for any~$V'$, $\varphi[V']$ is simply~$\varphi$ with a partial assignment.
    For the second part, we notice that if~$I_\alpha \subseteq V'$, then every value assigned in~$\varphi[V']$ is also assigned by~$\alpha$. If~$\ell\notin V'$ then~$\ell\notin I_\alpha$, and thus~$\alpha(\ell)=0$. 
\end{proof}

\begin{lemma}\label{lem:subsize}
    If~$V'\subseteq V(\mathcal{H}_\varphi)$ is of size~$|V'|\leq (1-\delta)|V(\mathcal{H}_\varphi)|$ for some~$\delta>0$, then the formula~$\varphi[V']$ has at most~$(1-2\delta)n$ unassigned variables.
\end{lemma}
\begin{proof}
    If there exits any~$x_i$ such that both~$x_i$ and~$\overline{x_i}$ are not in~$V'$ then~$\alpha[V']$ is a contradiction.
    Otherwise, for every~$i$ only one of~$x_i,\overline{x_i}$ can be missing from~$V'$ and thus if it is of size~$\leq (1-\delta)|V(\mathcal{H}_\varphi)|$, then at least~$\delta |V(\mathcal{H}_\varphi)| = 2\delta n$ literals are missing from~$V'$ and in particular are assigned in~$\varphi[V']$.
\end{proof}

We next define a family of hypegraphs we call~$(D,C,\varepsilon)$-structures. 
These are hypergraphs in which the edges are spread in a \emph{somewhat}-regular manner.
We would then present an improved~$k$-SAT algorithm for formulas~$\varphi$ such that~$\mathcal{H}_\varphi$ \emph{contains} any~$(D,C,\varepsilon)$-structure as a~\emph{subgraph} (not necessarily induced).
In Section~\ref{subsec:necess} we discuss this condition and show that formulas without such a structure as a subset are unlikely to have a faster algorithm than for general formulas. 

\begin{definition}\label{def:structure}
    Let~$\mathcal{H}$ by a~$r$-uniform hypergraph.
    We say that~$\mathcal{H}$ is a~$(D,C,\varepsilon)$-structure if:
    \begin{enumerate}
        \item $|E(\mathcal{H})|\geq D\cdot|V(\mathcal{H})|$.
        \item $\Delta_1(\mathcal{H}) \leq CD$.
        \item $\Delta_2(\mathcal{H}) \leq CD^{1-\varepsilon}$.
    \end{enumerate}
\end{definition}

We think of~$r,C,\varepsilon>0$ as fixed constants and of~$D$ as arbitrarily large.
We note that for example in a random hypegraph with~$D|V(\mathcal{H})|$ edges, the expected degree of a vertex is~$rD$ and the expected co-degree of a pair of vertices is~$\frac{r^2 D}{|V(\mathcal{H})|}<<D^{1-\varepsilon}$.
We now show that~$(D,C,\varepsilon)$-structures have good containers.

\begin{lemma}\label{lem:structureconts}
    For any~$r,C,\varepsilon>0$ there exist~$D_0\in\mathbb{N},\; \delta>0$ and a function~$r:\mathbb{N}\rightarrow[1,2]$ with~$\lim_{D\rightarrow \infty} r(D) = 1$, such that the following holds. 
    Let~$\mathcal{H}$ be a~$r$-uniform~$(D,C,\varepsilon)$-structure with~$D\geq D_0$.
    There exists a collection~$\mathcal{C}$ of subsets~$C_1,C_2,\ldots,C_r \subset V(\mathcal{H})$ such that:
    \begin{itemize}
        \item $r \leq r(d)^n$.
        \item For every~$i$,~$|C_i| < \left(1-\delta\right)n$.
        \item For every $I\in \mathcal{I}(\mathcal{H})$, there exists~$i\in [r]$ such that~$I\subseteq C_i$.
    \end{itemize}
    Furthermore, we can compute~$\mathcal{C}$ in~$O^*\left(|\mathcal{C}|\right)$ time.
\end{lemma}
\begin{proof}
We apply Theorem~\ref{thm:hyperconts} to~$\mathcal{H}$ with parameters~$r,C$ and~$p=D^{-\varepsilon/r}$.
We satisfy the Theorem's conditions as~$\Delta_1(\mathcal{H})\leq Cd \leq C\cdot p^0\frac{|E(\mathcal{H})|}{|V(\mathcal{H})|}$, and as for every~$2\leq i \leq r$
$$
\Delta_i(\mathcal{H})\leq \Delta_2(\mathcal{H}) \leq CD^{1-\varepsilon} \leq
CD^{-\varepsilon} \cdot \frac{|E(\mathcal{H})|}{|V(\mathcal{H})|}
=
Cp^r \cdot \frac{|E(\mathcal{H})|}{|V(\mathcal{H})|}
\leq
Cp^i \cdot \frac{|E(\mathcal{H})|}{|V(\mathcal{H})|}
.$$
When~$D\rightarrow \infty$, we have~$(r-1)p\rightarrow 0$.
Thus there exists~$D_0,r$ for which the statement holds.
\end{proof}

\begin{definition}
    We say that a~$k$-SAT formula~$\varphi$ \emph{contains} a~$(D,C,\varepsilon)$-structure if there is a subset~$E'\subseteq E(\mathcal{H}_\varphi)$ of~$\varphi$'s clauses such that~$(V(\mathcal{H}_\varphi),\;E')$ is a~$(D,C,\varepsilon)$-structure.
\end{definition}

\begin{theorem}\label{thm:ksatfinal}
    For every~$k,c,C,\varepsilon>0$ there exists~$D_0,\varepsilon'>0$ such that the following holds.
    If we have an algorithm for~$k$-SAT running in~$O(c^n)$ time, then we can solve~$k$-SAT for formulas containing a~$(D,C,\varepsilon)$-structure for any~$D\geq D_0$ in~$O\left(\left(c-\varepsilon'\right)^n\right)$ time. 
\end{theorem}
\begin{proof}
    We choose~$D_0$ later. Let~$E'$ be the~$(D,C,\varepsilon)$-structure contained in~$\varphi$.
    Using Lemma~\ref{lem:structureconts} on the subgraph~$\mathcal{H}':=(V(\mathcal{H}_\varphi),E')$, we get a collection~$\mathcal{C}$ of containers.
    We enumerate over the containers~$C_i\in \mathcal{C}$.
    For each container~$C_i$, we run the~$k$-SAT algorithm on the formula~$\varphi[C_i]$.
    Note that even though we computed the container~$C_i$ with respect to the subgraph~$\mathcal{H}'$, we consider the \emph{entire} formula~$\varphi$ restricted to the vertices of the container.
    If we found a satisfying assignment restricted to any of the containers, then we return it as~$\varphi$ is satisfiable by Lemma~\ref{lem:subsat}.
    If~$\varphi$ is satisfiable, then the independent set~$I_\alpha\in \mathcal{I}(\mathcal{H}_\varphi)$ corresponding to a satisfying assignment~$\alpha$ (as proven in Lemma~\ref{lem:satind}), is also an indepndent set in the subgraph~$\mathcal{H}'$ and in particular is fully contained in one of the containers~$C_i$. Thus,~$\varphi[C_i]$ is satisfiable by Lemma~\ref{lem:subsat}.

    By Lemma~\ref{lem:subsize}, each restricted formula we solve~$k$-SAT on contains at most~$(1-2\delta)n$ variables.
    Thus, the total running time is~$r(D)^n \cdot c^{(1-2\delta)n}$.
    We pick~$D_0$ large enough and~$\varepsilon'>0$ small enough so that~$r(D)c^{1-2\delta} < c-\varepsilon'$ for all~$D\geq D_0$.
\end{proof}

\subsection{On the necessity of~$(D,C,\varepsilon)$-structures}\label{subsec:necess}
In this section we discuss the reasoning and necessity for the three conditions in Definition~\ref{def:structure}.
We show that Conditions~$1$ and~$2$ are inherent, and that removing Condition~$3$ would imply a surprising result and it is thus plausible it is inherent as well.
By inherent we mean that if an improvement would be possible without any of these conditions, then it would also be possible in the completely general case.

\subsubsection*{Condition~$1$: Density}
This condition is the most natural one.
If we require~$\varphi$ to have many clauses, then clearly the average degree in~$\mathcal{H}_\varphi$ is high.

\subsubsection*{Condition~$2$: Concentration of Density}
This condition guarantees that the high number of clauses does not come from many clauses that are concentrated in a subset of variables of negligible size.
Consider a general, possibly sparse,~$k$-SAT formula~$\varphi$.
We may add~$\sqrt{n}$ new variables to~$\varphi$ with~$\Theta\left(\sqrt{n}^k\right)$ clauses on them, without changing the satisfiability of~$\varphi$.
The new formula~$\varphi'$ has many clauses, but is equivalent to~$\varphi$ and was computed efficiently. Thus, simply having many clauses cannot imply a faster algorithm than in the sparse case.
To avoid this type of an example, we would like the density of~$\varphi$ to not be concentrated. 
In particular, we want some~$\varepsilon>0,D$ to exist such that~$\varphi$ contains at least~$Dn$ clauses even after the removal of any subset of~$\varepsilon n$ variables.
This condition turns out to be equivalent to containing a~$(D',C',0)$-structure as a subset (for an appropriate choice of~$D',C'$).
\begin{lemma}\label{lem:structures1}
    Let~$\mathcal{H}$ be a~$r$-uniform hypergraph such that for every~$V'\subseteq V(\mathcal{H})$ of size~$|V'|\geq (1-\varepsilon)|V(\mathcal{H})|$ we have~$|E(\mathcal{H}[V'])|>D|V(\mathcal{H})|$. 
    Then,~$\mathcal{H}$ contains a~$\left(\frac{\varepsilon}{2} D,(r+1)D,0\right)$-structure as a subgraph.
\end{lemma}

\begin{lemma}\label{lem:structures2}
    Let~$\mathcal{H}$ be a~$r$-uniform hypergraph that contains a~$\left(D,C,0\right)$-structure as a subgraph.   
    Then, for every~$V'\subseteq V(\mathcal{H})$ of size~$|V'|\geq (1-\varepsilon)|V(\mathcal{H})|$ we have~$|E(\mathcal{H}[V'])|>\left(D-\varepsilon C D\right)|V(\mathcal{H})|$. 
\end{lemma}

The proofs of Lemmas~\ref{lem:structures1} and~\ref{lem:structures2} are deferred to Appendix~\ref{appendix:equivdensity}.

\subsubsection*{Condition 3: Multiplicity}
This condition guarantees that the high number of clauses does not come from a few clauses that appear with high multiplicity.
Clearly, we can arbitrarily increase the number of clauses in any formula by simply duplicating them.
Hence, we should at least require~$\mathcal{H}$ to be simple (i.e.,~$\Delta_r(\mathcal{H})\leq 1$).
Could that be enough?
Let~$\varphi$ be a~$k$-SAT formula, and assume it contains a clause~$C$ that is strictly shorter than~$k$. That is,~$C$ contains~$i<k$ literals.
Without changing its satisfiabilty, we may \emph{add} to~$\varphi$ any extension of~$C$, this is a clause containing~$C$ and any other~$k-i$ literals.
Thus, any~$i<k$ clause can be replaced with arbitrarily many clauses without changing~$\varphi$'s satisfiability and without creating a multi-edge. On the other hand, this increases~$\Delta_i(\mathcal{H})$ and thus we add the requirement that it is much smaller than~$D$ for every~$i>1$.

We note that it is possible that in this case better algorithms exist for another reason; It is not clear whether containing a non-negligible subset of shorter clauses benefits~$k$-SAT algorithms.
We thus pose the following problem.
\begin{problem}[Informal]\label{problem:ksatmixed}
    Let~$\varphi$ be a (possibly sparse)~$k$-CNF formula, and assume that a constant fraction of its clauses are of size strictly smaller than~$k$. Can we solve~$k$-SAT on~$\varphi$ more quickly than for general~$k$-SAT formulas?
\end{problem}
Removing Condition~$3$ would imply a positive answer for this problem.

\section{Conclusions and Open Problems}\label{sec:conc}
The main purpose of this paper is to demonstrate how the tool of Hypergraph Containers can be used algorithmically. 
We show it provides a black-box or white-box improvement for several algorithms if the input is assumed to be almost-regular.
Being almost-regular is a seemingly weak condition, that generalizes two classic families of inputs for which we usually see improved algorithms: sparse inputs and random inputs.

A satisfying corollary of this is that a couple of the most prominent constraint satisfaction problems,~$k$-SAT and Graph Coloring, indeed have faster algorithms for very dense inputs.
This formalizes for the first time a natural intuition that many constraints should make these problems easier.

We believe that the tools presented in this paper should be relevant to many other problems, and thus the main question we pose is which other problems can benefit from these?

Each of our two main applications also leaves a natural open problem.
For graph coloring, it is still open whether~$k$-coloring can be solved in~$(2-\varepsilon)^n$ time, without any restrictions on the input graphs.
For dense~$k$-SAT formulas, Problem~\ref{problem:ksatmixed} and the question of whether or not Condition~$3$ is necessary is interesting.
For both of those questions, it is likely that new problem-specific (and not black-box) tools are necessary.

Throughout this paper, we did not try to optimize the magnitude of improvement we achieve.
Thus, it is left open to understand and optimize the results of this paper quantitatively.
In~\cite{balogh2019efficient}, Balogh and Samotij give a much more efficient (in terms of parameters) version of the hypergraph container lemma. This version can be used for quantitative improvements.

Another intriguing question is whether the Partition Containers that are introduced in Section~\ref{sec:partcont} are useful in combinatorics, and not only for algorithmic purposes.

\bibliographystyle{alpha}
\bibliography{arxiv/main}

\appendix

\section{Appendix}

This section contains proofs that are deferred from the main body of the paper.
All of the results proven in this section are \textbf{not} directly used in the proofs of any of the main results in the paper.

\subsection{Algorithm for Extensions-Sum with~$k=3$}\label{appendix:es3}
\begin{proof}[Proof of Lemma~\ref{extsum3}]
    Similarly to the proof of Lemma~\ref{extsum2}, we externally enumerate over all assignments of~$X_\cap := X_1\cap X_2\cap X_3$.
    For each such assignment, we construct a weighted complete tripartite graph as follows.
    The vertices in the first part correspond to all assignments to~$\left(X_1\cap X_2\right)\setminus X_\cap$, in the second to~$\left(X_2\cap X_3\right)\setminus X_\cap$, and in the third to~$\left(X_1\cap X_3\right)\setminus X_\cap$.
    Consider an edge from a vertex corresponding to an assignment of~$\left(X_1\cap X_3\right)\setminus X_\cap$ to one corresponding to an assignment of~$\left(X_1\cap X_2\right)\setminus X_\cap$.
    Together, they correspond to an assignment~$\alpha$ of~$X_1 \cap \left(X_2\cup X_3\right)$, we set the weight of this edge to be~$\sum_{\alpha':X_1\setminus\left(X_2\cup X_3\right)\rightarrow\{0,1\}} f_1(\alpha,\alpha')$.
    The weights of edges in the two other parts are set symmetrically.
    The sum over all triangles in the graph of their weights (i.e., product of triangle edge weights), is exactly the Extensions-Sum output.
    It is left to note that if~$a\leq b\leq c$ then the weighted sum of all triangles in a tripartite graph with parts of sizes~$(a,b,c)$ can be computed in~$\frac{b}{a} \cdot \frac{c}{a} \cdot a^\omega = a^{\omega-2} b c$ time using fast matrix multiplication.
\end{proof}

\subsection{Extensions-Sum and Hypercliques}\label{appendix:hypercliques}
\begin{proof}[Proof of Lemma~\ref{extsum_lb}]
    Let~$\mathcal{H}$ be a~$r$-uniform hypegraph on~$n$ vertices.
    Let~$Y_1,\ldots,Y_k$ be disjoint sets of size~$\lceil \log n \rceil$ each.
    For each~$1\leq i\leq k$, we arbitrarily choose a bijection~$m_i : \left(Y_i \rightarrow \{0,1\}\right) \rightarrow \left[2^{\lceil\log n\rceil}\right]$.
    Our set of variables is~$X=\bigcup_{i=1}^{k} Y_i$.
    For each~$I\in {{[k]} \choose r}$ we define a subset~$X_I = \bigcup_{i\in I} Y_i$, and a function~$f_I : X_I \rightarrow \{0,1\}$ as follows.
    $f_I(\alpha)$ is~$1$ if~$\{m_i\left(\alpha\vert_{Y_i}\right)\;\mid\;i\in I\}$ is an hyperedge in~$\mathcal{H}$ and~$0$ otherwise.
    It is straighforward to verify that~$\prod_{I} f_I(\alpha)=1$ if~$\{m_i\left(\alpha\vert_{Y_i}\right)\;\mid\;i\in [k]\}$ is a clique, and~$0$ otherwise.
\end{proof}

\subsection{Non-Covering Partitions}\label{appendix:partitions}
\begin{lemma}
    For any~$k$ there exist~$K$ and a~$\left(K,\frac{1}{2}\right)$-collection without a~$(k,1-\varepsilon)$-refinement for any~$\varepsilon>0$.
\end{lemma}
\begin{proof}
    Let~$K={2k\choose k}$.
    %We first build a collection of subsets of a set of size~$2k$.
    Let~$X=[2k]$ and our subsets be~$\mathcal{X}={{[2k]} \choose k}$ be all subsets of~$X$ of size~$k$.
    Let~$P_1,\ldots,P_k$ be a refinement of~$\mathcal{X}$ of size~$k$.
    If for every~$1\leq i\leq k$ we have~$P_i\neq X$, then pick arbitrarily~$x_i \in X \setminus P_i$ for every~$i$.
    The subset~$\{x_1,\ldots,x_k\} \in \mathcal{X}$ cannot be contained in any part of the partition, which is a contradiction.

    We note that we can blow-up the size of the set by taking arbitrarily many copies of it.
\end{proof}

We next note that the above size of~$K$ is essentially the best possible.
\begin{lemma}
    Every~$\left(K,\frac{1}{2}\right)$-collection in~$X$ has a~$\left(k,1-\frac{1}{|X|}\right)$-refinement for~$k=\lfloor\log K + 1\rfloor$.
\end{lemma}
\begin{proof}
    Let~$x_1,\ldots,x_k$ be picked uniformly in random out of~$X$.
    The probability that~$\{x_1,\ldots,x_k\}\subseteq X_i$ for a specific subset~$X_i$ in the collection is at most~$2^{-k}$.
    Thus, the expected number of subsets~$X_i$ for which~$\{x_1,\ldots,x_k\}\subseteq X_i$ is at most~$K2^{-k}<1$.
    In particular, there exist a choice of~$x_1,\ldots,x_k$ for which $\{x_1,\ldots,x_k\}\setminus X_i$ is not empty for any subset~$X_i$.
    We define a~$k$-partition by taking the union of all subsets not containing~$x_i$ as the~$i$-th part.
\end{proof}

\subsection{Equivalence between notions of $k$-SAT density}\label{appendix:equivdensity}
\begin{proof}[Proof of Lemma~\ref{lem:structures1}]
    We construct a subgraph iteratively.
    We begin with~$E'=\emptyset$ and~$R=\emptyset$.
    As long as there exists a vertex~$v\in V(\mathcal{H})$ with degree at least~$D$ in~$\mathcal{H}[V(\mathcal{H})\setminus R]$, we add~$D$ of its adjacent edges to~$E'$.
    We then add~$v$ itself to~$R$.
    We also add to~$R$ every vertex in~$V(\mathcal{H})$ of degree higher than~$rD$ in~$(V(\mathcal{H}), E')$.

    Vertices are added to~$R$ in two manners in each iteration. First, the vertex~$v$ itself. Second, all vertices for which the degree in~$E'$ got too large.
    We note that the number of vertices of the second type is at most the number of vertices of the first type. In every iteration we increase the sum of degrees in~$E'$ by exactly~$rD$, and each vertex removed is of degree at least~$rD$.

    Finally, we note that as long as~$R\leq \varepsilon |V(\mathcal{H})|$ then another iteration is possible as the average degree in~$\mathcal{H}[V(\mathcal{H})\setminus R]$ is at least~$D$.
    Hence, we run the algorithm for at least~$\frac{\varepsilon}{2} |V(\mathcal{H})|$ iterations, and add at least~$\frac{\varepsilon D}{2} |V(\mathcal{H})|$ edges to~$E'$.

    As we add vertices to~$R$ if they reach degree~$rD$ in~$E'$, and as in each iteration a vertex's degree can only increase by~$D$, the maximum degree in~$E'$ is~$(r+1)D$.
\end{proof}
\begin{proof}[Proof of Lemma~\ref{lem:structures2}]
    Straightforward by considering only edges of the structure.
\end{proof}

\end{document}